\documentclass[12pt]{article}
\usepackage{amssymb}
\usepackage{amsmath}
\usepackage{amsthm}
\usepackage[space]{cite}
\usepackage{algorithm}
\usepackage{bbm}
\usepackage[noend]{algpseudocode}
\usepackage[normalem]{ulem}
\usepackage{url}
\usepackage{hyperref}
\usepackage{graphicx}
\usepackage{adjustbox}
\usepackage{caption}
\usepackage{subcaption}
\usepackage{enumitem}
\usepackage{xcolor}
\usepackage{comment}
\usepackage{blkarray}
\usepackage{multirow}
\usepackage{booktabs}
\usepackage{MnSymbol}

\usepackage{cases}

\numberwithin{equation}{section}

\usepackage[symbol]{footmisc}

\newcommand{\id}{\mathbb{I}}

\newcommand{\cut}[1]{\ifmmode\text{\textcolor{red}{\sout{\ensuremath{#1}}}}\else\textcolor{red}{\sout{#1}}\fi}

\newcommand{\RN}[1]{%
  \textup{\uppercase\expandafter{\romannumeral#1}}%
}

\usepackage{quantikz}
\usetikzlibrary{circuits.logic.US}
\tikzset{
    gateO/.style={
        draw,
        circle,
        minimum width=0.5em,
        inner sep=2pt    }
}
\DeclareExpandableDocumentCommand{\gateO}{O{}{m}}{|[gateO,#1]| {#2} \qw}

\tikzset{
    gateOS/.style={
        draw,
        circle,
        minimum width=0.5em,
        inner sep=2pt,
		fill=red!20}
}
\DeclareExpandableDocumentCommand{\gateOS}{O{}{m}}{|[gateOS,#1]| {#2} \qw}

\newtheorem*{lemma*}{Lemma}
\newtheorem*{prop*}{Proposition}

\begin{document}

\begin{titlepage}

\begin{center}
\vspace*{\skip5}

{\LARGE Simple ways of preparing qudit Dicke states}\\
\vspace{1.cm}

{\large Noah B. Kerzner$^{1, 2}$, Federico Galeazzi$^{2,3}$ and Rafael I. Nepomechie$^2$}\footnote[1]{Corresponding author: nepomechie@miami.edu}

\vspace{1.cm}
$^1$Department of Physics and Astronomy, 153 Olin Science Bldg.\\
Bucknell University, Lewisburg, PA 17837 USA\\
%{\tt nbk005@bucknell.edu}\\
\vspace{0.5cm}
$^2$Department of Physics, PO Box 248046\\
University of Miami, Coral Gables, FL 33124 USA\\ 
%{\tt nepomechie@miami.edu}\\ 
\vspace{0.5cm}
$^3$Coral Gables High School, 450 Bird Rd\\ 
Coral Gables, FL 33146 USA\\

\end{center}

\vspace{.5in}

\begin{abstract}
Dicke states are permutation-invariant superpositions of qubit computational basis states, which play a prominent role in quantum information science. We consider here two higher-dimensional generalizations of these states: $SU(2)$ spin-$s$ Dicke states and $SU(d)$ Dicke states. We present various ways of preparing both types of qudit Dicke states on a qudit quantum computer, using two main approaches: a deterministic approach, based on exact canonical matrix product state representations; and a probabilistic approach, based on quantum phase estimation. The quantum circuits are explicit and straightforward, and are arguably simpler than those previously reported.
\end{abstract}

\end{titlepage}

\setcounter{footnote}{0}

\section{Introduction}\label{sec:intro}

Dicke states \cite{Dicke:1954zz} are permutation-invariant superpositions of qubit computational basis states, for example,
\begin{equation}
|D_{3,2}\rangle =\frac{1}{\sqrt{3}}\left(|0 1 1 \rangle + |1 0 1  \rangle +|1 1 0 \rangle \right) \,. 
\end{equation}
These states play a prominent role in quantum information science, see e.g. the recent review \cite{Marconi:2025ioa}.  Considerable effort has been devoted recently to the preparation of these states on a qubit quantum computer, see e.g. \cite{Bartschi2019, Wang:2021, Buhrman:2023rft, Bond:2023yry, Piroli:2024ckr, Yu:2024szp, Liu:2024taj}.

Considerable attention has also been devoted recently to extending qubit-based quantum computing to higher dimensions, both theoretically \cite{Wang:2020} and experimentally \cite{Ringbauer:2022, Goss:2022, Luo:2022pxs, Lindon:2022ekh, Fischer:2022dbr, Liu:2023otw, Meng:2023uor, Brock:2024vkc}. It is therefore natural to consider higher-dimensional (qudit) generalizations of Dicke states. Two such generalizations are $SU(2)$ spin-$s$ Dicke states and $SU(d)$ Dicke states. An example of the former with $s=1$ (and therefore $2s+1=3$ levels) is
\begin{equation}
    |D^{(1)}_{3, 2}\rangle = \frac{2}{\sqrt{15}}\left(|0 1 1 \rangle + |1 0 1  \rangle +|1 1 0 \rangle \right) +  \frac{1}{\sqrt{15}}\left(|0 0 2 \rangle + |0 2 0  \rangle +|2 0 0 \rangle \right) \,,
    \label{spinsexample}
\end{equation}
with a fixed digit sum (here, 2) in each basis state. An example of the latter with $d=3$ is 
\begin{equation}
    |D^{3}(1,1,1)\rangle = \frac{1}{\sqrt{6}}\left( |0 1 2 \rangle 
    + |0 2 1  \rangle + |1 0 2 \rangle + |1 2 0 \rangle + |2 0 1 \rangle
    + |2 1 0 \rangle
    \right) \,,
    \label{sudexample}
\end{equation}
with one 0, one 1, and one 2 in each basis state. More precise definitions and explanations of notation can be found at the beginning of Secs. \ref{sec:spinsDicke} and \ref{sec:sudDicke}, respectively. Potential applications of such states include quantum error correction \cite{Herbert:2023qgu, Uy:2024ree}, metrology  \cite{Lin:2024utr}, and quantum interferometric imaging \cite{Liu:2024taj}.

Much less effort has been devoted to the preparation of qudit Dicke states than to ordinary (qubit) Dicke states. For $SU(2)$ spin-$s$ Dicke states, a recursive state-preparation algorithm, generalizing the one by B\"artschi and Eidenbenz \cite{Bartschi2019} for qubit ($s=1/2$) Dicke states, was formulated in \cite{Nepomechie:2024fhm}. For $SU(d)$ Dicke states, a recursive state-preparation algorithm was formulated in \cite{Nepomechie:2023lge}; and an algorithm using sorting networks has recently been formulated by Liu, Childs and Gottesman \cite{Liu:2024taj}.

Here we consider further the problem of preparing both types of generalized Dicke states on a qudit quantum computer. We formulate quantum circuits implementing the sequential deterministic preparation \cite{Schon:2005}
of these states based on their recently-found exact matrix product state (MPS) representations \cite{Raveh:2024sku}. We also formulate circuits based on quantum phase estimation (QPE) \cite{Kitaev:1995qy, Nielsen:2019} that prepare these states probabilistically, some of which achieve constant circuit depth, generalizing an approach developed in \cite{Wang:2021, Piroli:2024ckr} for preparing ordinary Dicke states. For simplicity, we restrict our attention to exact state preparation; computational resources can be reduced by dropping this requirement \cite{Piroli:2024ckr}.

The remainder of this paper is organized as follows. In Sec. \ref{sec:spinsDicke}, we consider the preparation of $SU(2)$ spin-$s$ Dicke states, starting with the sequential preparation, and then continuing with the QPE preparation. In Sec. \ref{sec:sudDicke}, we consider the preparation of $SU(d)$ Dicke states, again starting with the sequential preparation, and continuing with the QPE preparation. We briefly discuss these results Sec. \ref{sec:discussion}. Implementations
in cirq \cite{cirq} of all the circuits are available on GitHub \cite{GitHub}.

A summary of our main results, and a comparison with previous work, are presented in Table \ref{table:summary}. 

\begin{table}[ht]
\centering
\begin{tabular}{|c||c|c|c|c|}
\hline\hline
\textbf{Dicke state}
& \textbf{Reference} 
& \textbf{Depth} 
& \textbf{Ancillas [Dimension] }
& \textbf{Repetitions} \\
\hline
\multirow{6}{*}{\shortstack{$SU(2)\ \text{spin-}s$\\ \vspace{0.2cm}
$|D^{(s)}_{n,k}\rangle$}}
     & NRR \cite{Nepomechie:2024fhm}& $\mathcal{O}(s k n)$ & 0 & 1 \\ \cline{2-5}
    & Result \hyperlink{res: 1}{\textbf{1}} & $\mathcal{O}(s k n)$ & 1 $[k+1]$ & 1 \\ \cline{2-5}
    & Result \hyperlink{res: 2}{\textbf{2}} & $\mathcal{O}(\log(s n))$ &  $\mathcal{O}(\log(s n) + n)$ $[2]$ 
    & $\mathcal{O}(\sqrt{s n})$ \\ \cline{2-5}
    & Result \hyperlink{res: 3}{\textbf{3}} & $\mathcal{O}(1)$ & $\mathcal{O}(n)$ $[2sn+1]$ 
    & $\mathcal{O}(\sqrt{s n})$  \\  \cline{2-5} 
    & \multirow{2}{*}{Result \hyperlink{res: 4}{\textbf{4}}} 
    & \multirow{2}{*}{$\mathcal{O}(1)$} 
    & \multirow{2}{*}{\shortstack{$\mathcal{O}(\log(s n))$ $[2]$ \,,\\
    $\mathcal{O}(n\log(s n))$ $[2s+1]$}} 
    & \multirow{2}{*}{$\mathcal{O}(\sqrt{s n})$} \\ 
    & & & & \\ %\cline{2-5} 
\hline \hline
\multirow{7}{*}{\shortstack{$SU(d)$\\ \vspace{0.2cm}
$|D^n(\vec{k})\rangle$}} 
    & NR \cite{Nepomechie:2023lge} & $\mathcal{O}(n^d)$ & 0 & 1 \\ \cline{2-5}
    & LCG \cite{Liu:2024taj} & $\mathcal{O}(\log n )$ & $\mathcal{O}(n\log n + \log d)$ $[2]$ & $\mathcal{O}(1)$ \\ \cline{2-5}
    & \multirow{2}{*}{Result \hyperlink{res: 5}{\textbf{5}}}
    & \multirow{2}{*}{$\mathcal{O}((n/d)^d)$}  
    & \multirow{2}{*}{\shortstack{1 $[2]$ \,, \\
    1 $[ \mathcal{O}((n/d)^d)]$}} 
    &  \multirow{2}{*}{1} \\ 
    & & & & \\ \cline{2-5}
    & Result \hyperlink{res: 6}{\textbf{6}} & $\mathcal{O}(d \log n)$  & $\mathcal{O}(d\log n + n)$ $[2]$
    & $\mathcal{O}(n^{(d-1)/2})$ \\ \cline{2-5}
    & Result \hyperlink{res: 7}{\textbf{7}} & $\mathcal{O}(d)$ & $\mathcal{O}(n+d)$ $[n+1]$ 
    & $\mathcal{O}(n^{(d-1)/2})$ \\ \cline{2-5}
    & \multirow{2}{*}{Result \hyperlink{res: 8}{\textbf{8}}} 
    & \multirow{2}{*}{$\mathcal{O}(1)$}
    & \multirow{2}{*}{\shortstack{$\mathcal{O}(d \log n)$ $[2]$ \,, \\
    $\mathcal{O}(n d\log n)$ $[d]$}}
    & \multirow{2}{*}{$\mathcal{O}(n^{(d-1)/2})$} \\ %cline{2-5}
    & & & & \\
\hline\hline
\end{tabular}
\caption{Summary of our results and comparison with previous work.
Note that we report here worst-case values, corresponding to $k \sim s n$ and $\vec{k} \sim (n/d, \ldots, n/d)$ for $|D^{(s)}_{n,k}\rangle$ and $|D^n(\vec{k})\rangle$, respectively; see the respective sections for more comprehensive discussions.}
\label{table:summary}
\end{table}

\section{$SU(2)$ spin-$s$ Dicke states $|D^{(s)}_{n,k}\rangle$}\label{sec:spinsDicke}

\par The ordinary qubit Dicke state can be written as $|D_{n,k}\rangle \propto (\mathbb{S^-})^k|0\rangle^{\otimes n}$, where $\mathbb{S^-}$ is the total spin-lowering operator for a system of $n$ spin-1/2 spins (qubits), which is applied $k$ times on the product state $|0\rangle^{\otimes n}$. A natural higher-dimensional generalization is to consider instead spin-$s$ spins (that is, $(2s+1)$-level qudits, where $s=1/2, 1, 3/2, \ldots$), and define the normalized state $|D^{(s)}_{n,k}\rangle \propto (\mathbb{S^-})^k|0\rangle ^{\otimes n}$, where $\mathbb{S^-}$ is now the total spin-lowering operator for a system of $n$ such spin-$s$ spins. An example with $s=1\,, n=3$ and $k=2$ is given in \eqref{spinsexample}, where $|0\rangle=(1,0, \ldots, 0)^T, \ldots, |2s\rangle= (0, \ldots, 0, 1)^T$ are the standard computational basis states in 
the $(2s+1)$-dimensional complex vector space
$\mathbb{C}^{2s+1}$, and tensor products are understood e.g. $|002\rangle = |0\rangle \otimes|0\rangle \otimes |2\rangle$. Note that the digit sum in each basis state in \eqref{spinsexample} is $k=2$.

Spin-$s$ Dicke states have the closed-form expression  \cite{Nepomechie:2024fhm}
\begin{equation}
|D^{(s)}_{n,k}\rangle = 
\sum_{\substack{m_i = 0, 1, \ldots, 2s \\
m_1 + m_2 + \cdots + m_n = k}} 
\sqrt{
\frac{
\binom{2s}{m_1} \binom{2s}{m_2} \cdots \binom{2s}{m_n}
}{
\binom{2sn}{k}
}
}
\; |m_n \ldots m_2\, m_1\rangle \,, \qquad k = 0, 1, \ldots, 2 s n \,,
\label{spin-s_Dicke_def}
\end{equation}
where $\binom{n}{m}= n!/(m! (n-m)!)$ is the binomial coefficient.
These states are $U(1)$ eigenstates for any allowed value of $s$
\begin{equation}
\mathbb{K}\, |D^{(s)}_{n,k}\rangle = k\, |D^{(s)}_{n,k}\rangle \,,
\label{U1spins}
\end{equation}
 where $\mathbb{K}$ is a Hermitian operator defined by
\begin{equation}
    \mathbb{K} = \sum_{i=1}^n \mathbbm{k}_i \,, \qquad 
    \mathbbm{k}_i = \id \otimes \ldots \id \otimes \overset{i}{\overset{\downarrow}{\mathbbm{k}}}\otimes \id \ldots \otimes \id \,,   
    \qquad 
    \mathbbm{k} = \sum_{m=0}^{2s} m |m\rangle\langle m|\,,   
    \qquad 
    \id = \sum_{m=0}^{2s} |m\rangle\langle m| \,.
    \label{Kdef}
\end{equation}
Hence, $\mathbb{K} = n s - \mathbb{S}^z$, where
$\mathbb{S}^z$ is the $z$-component of the total spin $\vec{\mathbb{S}}$. These states also have the ``duality'' (charge conjugation) property \cite{Nepomechie:2024fhm}
\begin{equation}
    \mathcal{C}^{\otimes n}\,  |D^{(s)}_{n,k}\rangle =  
    |D^{(s)}_{n, 2 s n - k}\rangle \,, \qquad  
    \mathcal{C} =  \sum_{m=0}^{2s} |2s-m\rangle\langle m|\,,
    \label{duality}
\end{equation}
which maps $k \mapsto 2sn-k$.

\par A recursive deterministic approach for preparing spin-$s$ Dicke states was presented in \cite{Nepomechie:2024fhm}.
In this section we present several additional methods of preparing these states: a sequential deterministic approach in Sec. \ref{sec:spinsDickeSeq}, and a probabilistic approach based on QPE in Sec. \ref{sec:spinsDickeQPE}.

\subsection{Sequential preparation}\label{sec:spinsDickeSeq}

\par An exact canonical MPS representation for 
$|D^{(s)}_{n,k}\rangle$ with minimal bond dimension $\chi=k+1$ is given by \cite{Raveh:2024sku}  
\begin{equation}
    |D^{(s)}_{n,k}\rangle = \sum_{\vec m}\langle\underline{k}|A^{m_n}_n\ldots A^{m_2}_2 A^{m_1}_1|\underline{0}\rangle|\vec m\rangle \,,
    \label{MPSspins}
\end{equation}
where $|\vec m\rangle= |m_n \ldots m_2\, m_1\rangle$ as in \eqref{spin-s_Dicke_def}; moreover, $|\underline{0}\rangle, \ldots, |\underline{k}\rangle$ are basis states of an ancilla qudit of dimension $\chi$ (an underline is used here to distinguish $\chi$-dimensional vectors from $(2s+1)$-dimensional vectors), and $A^m_i$ are $(k+1) \times (k+1)$ matrices with elements (for $0 < k < s n$)
\begin{equation}
\langle\underline j'|A_i^{m}|\underline j\rangle=\gamma^{(i)}_{j,m}\,\delta_{j',j+m} \,, \qquad
\gamma^{(i)}_{j,m} = \sqrt{ \dfrac{ \binom{2s(n - i)}{k - j - m} \binom{2s}{m} }{ \binom{2s(n - i + 1)}{k - j} } }\,.
\label{Dickeelemspins}
\end{equation}
Note that $\gamma^{(i)}_{j,m}=0$ if $k - j > 2s(n - i + 1)$, and that the binomial coefficient $\binom{n}{m}$ is defined to be zero if $m \notin [0\,, n]$.

The fact that the MPS is canonical implies \cite{Schon:2005} that we can define a two-qudit unitary operator $U_i$ acting on an ancilla qudit $|\underline{j}\rangle$ and a ``system'' qudit at site $i$ (where $i=1, 2, \ldots, n$) that performs the mapping 
\begin{equation}
    U_i|\underline{j}\rangle|0\rangle_i=\sum_m(A^m_i |\underline{j}\rangle)|m\rangle_i=\sum_{m=0}^{2s} \gamma^{(i)}_{j,m}|\underline{j+m}\rangle|m\rangle_i \,,
    \label{u-gate-su(2)}
\end{equation}
where the second equality follows from \eqref{Dickeelemspins}.
In view of the first equality in \eqref{u-gate-su(2)}
and \eqref{MPSspins}, the state $|D^{(s)}_{n,k}\rangle$
can be prepared sequentially as follows
\begin{equation}
    |\underline{k}\rangle|D^{(s)}_{n,k}\rangle=U_n\ldots U_2U_1|\underline{0}\rangle|0\rangle^{\otimes n} \,.
    \label{sequentialspins}
\end{equation}

In order to implement the sequential preparation \eqref{sequentialspins}, it is necessary to explicitly implement the unitary $U_i$ \eqref{u-gate-su(2)}, whose coefficients $\gamma^{(i)}_{j,m}$ depend on the state $|\underline{j}\rangle$ of the ancilla qudit. To  this end, we decompose $U_i$ into an ordered product of simpler operators ${\rm I}^{(i)}_l$
\begin{equation}
    U_i=\overset{\curvearrowleft}{\prod_{l = \max\left(0,\ 2s(i - n - 1) + k\right)}^{\min\left(2si,\ k\right) - 1}} 
    {\rm I}^{(i)}_l \,,
 \label{spinsU}
\end{equation}
where the product goes from right to left with increasing $l$, 
such that 
\begin{align}
    {\rm I}^{(i)}_l|\underline{j}\rangle|0\rangle_i &=
    \begin{cases}
       \sum_{m=0}^{2s} \gamma^{(i)}_{j,m}|\underline{j+m}\rangle|m\rangle_i 
       & \text{if} \hspace{.2cm } l=j  \\ |\underline{j}\rangle|0\rangle_i 
       & \text{if} \hspace{.2cm } l\ne j
    \end{cases} \,,  \label{eq:spinsIa} \\
{\rm I}^{(i)}_l|\underline{j}\rangle|m\rangle_i &=|\underline{j}\rangle|m\rangle_i  \quad  \text{for }m>0 \quad \text{and} \quad j \le l+m-1 \,. 
    \label{eq:spinsIb}
\end{align}
The latter condition \eqref{eq:spinsIb} ensures that these gates do not interfere
\begin{equation}
     {\rm I}^{(i)}_l \left(  {\rm I}^{(i)}_j |\underline{j}\rangle|0\rangle_i \right)
     =  \left(  {\rm I}^{(i)}_j |\underline{j}\rangle|0\rangle_i \right) \quad \text{for} \quad l > j \,. 
\end{equation}
Although {\it a priori} all ${\rm I}^{(i)}_l$ operators from $l=0$ to $l=k$ could contribute to $U_i$, one can check that only those operators in \eqref{spinsU} can act non-trivially.

The operators ${\rm I}^{(i)}_l$ can be implemented by the quantum circuit
whose circuit diagram is shown in Figure \ref{fig:I_spin-s}. The top wire is a ``system'' qudit of dimension $2s+1$, and the bottom wire is the MPS ancilla qudit of dimension $\chi = k+1$.  The circle $\begin{quantikz}\gateO{\scriptscriptstyle i}\end{quantikz}$ denotes a control on the value $i$. The 1-qudit gate ${\rm X}_d$ and the 2-qudit controlled gate 
${\rm SUM}_d$ (also known as a fan-out gate)
are defined as (see e.g. \cite{Wang:2020}; we remark that 1-qudit and 2-qudit gates are already becoming experimentally feasible, see e.g. \cite{Ringbauer:2022, Goss:2022, Luo:2022pxs, Lindon:2022ekh, Fischer:2022dbr, Liu:2023otw, Meng:2023uor, Brock:2024vkc})
\begin{equation}
    {\rm X}_d\, |x\rangle = |x+1 \rangle\,, \qquad 
    {\rm X}_d^\dagger\, |x\rangle = |x-1 \rangle \,,
\end{equation}
and
\begin{equation}
    {\rm SUM}_d\, |y\rangle|x\rangle = |y + x \rangle|x\rangle\,, \qquad 
    {\rm SUM}_d^\dagger\, |y\rangle|x\rangle = |y - x\rangle|x \rangle \,,
\end{equation}
respectively, where the sums are defined modulo $d$, and here $d=\chi=k+1$. (We use $\begin{quantikz}\gateO{\scriptscriptstyle *}\end{quantikz}$ to denote the control of the ${\rm SUM}_d$ gate, since 
the control qudit can have any value.)
The rotation gate $R_{m,m+1}(\theta_m)$ is defined as
\begin{equation}
R_{m,m+1}(\theta_m) = \exp\left(-\frac{\theta_m}{2} \, (|m\rangle\langle m+1| - |m+1\rangle\langle m|)\right)  \,,
\label{rotation gate}
\end{equation}
where $\theta_m$ is given by
\begin{equation}
     \theta_m= 2 \arccos\left( \frac{\gamma^{(i)}_{l,m}}
     {\prod_{p=0}^{m-1}\sin(\theta_p/2)} \right) \,,
     \qquad m = 0, 1, \ldots, 2s-1\,.
     \label{thetaspins}
\end{equation}
Each rotation gate is controlled by the ancilla value $l+1$, and all operations are assumed to be modulo $\chi$. It is straightforward to check that this circuit satisfies the properties \eqref{eq:spinsIa}, \eqref{eq:spinsIb}. For  $s=1/2$, this circuit reduces to the one presented in the appendix of \cite{Raveh:2024sku}.

\begin{figure}[H]
  \centering
  \scalebox{0.95}{
  \begin{quantikz}[row sep=0.6cm, column sep=0.2cm]
    \lstick{$|x\rangle_i$} 
      & \qw
      & \gateO{*} \vqw{1}
      & \gate{R_{0,1}(\theta_0)} \vqw{1} 
      & \gate{R_{1,2}(\theta_1)} \vqw{1} 
      & \ldots \ldots
      & \gate{R_{2s-1,2s}(\theta_{2s-1})} \vqw{1} 
      & \gateO{*} \vqw{1} 
      & \qw
      & \qw \\
    \lstick{$|\underline{j}\rangle$} 
      & \gate{{\rm X}_{\chi}}
      & \gate{{\rm SUM}_{\chi}^\dagger} 
      & \gateO{{\scriptscriptstyle l+1}} 
      & \gateO{{\scriptscriptstyle l+1}} 
      & \ldots \ldots
      & \gateO{{\scriptscriptstyle l+1}} 
      & \gate{{\rm SUM}_{\chi}}
      & \gate{{\rm X}_{\chi}^{\dagger}}
      & \qw
  \end{quantikz}
  }
  \caption{Circuit diagram for ${\rm I}^{(i)}_l$ defined in \eqref{eq:spinsIa}, \eqref{eq:spinsIb}.}
  \label{fig:I_spin-s}
\end{figure}
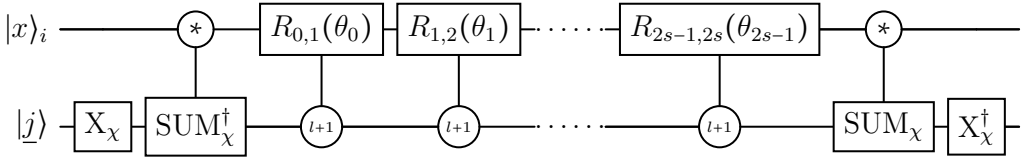

The complete circuit diagram for preparing the state $|D^{(s)}_{n,k}\rangle$ sequentially \eqref{sequentialspins} is shown in Figure \ref{fig:Dicke}. We therefore have the following:

\hypertarget{res: 1}{\noindent{\bf Result 1.}} The state $|D^{(s)}_{n,k}\rangle$ can be prepared deterministically with approximate depth 
$\mathcal{O}(s k n)$ for $k \sim s n$, using one ancilla of dimension $k+1$.

\noindent
This circuit depth is comparable to that of the circuit in \cite{Nepomechie:2024fhm}. For $k\ll s n$, the depth is lower due to the restriction on the $l$-values in the product \eqref{spinsU}.

\begin{figure}[H]
	\centering
	\begin{subfigure}{0.60\textwidth}
      \centering
\begin{adjustbox}{width=1.0\textwidth, raise=2em}
\begin{quantikz}
\lstick{$i$} & \gate[style={fill=red!20}]{} \vqw{1} & \qw \rstick[2, brackets=none]{$\quad\equiv\quad$}\\
\lstick{} & \gate[style={fill=blue!20}]{}  & \qw \\
\end{quantikz}
\begin{quantikz}
\lstick{$i$} & \gate[2]{I^{(i)}_x} & \gate[2]{I^{(i)}_{x+1}} & \quad\ldots\quad
& \gate[2]{I^{(i)}_y} & \qw \\
\lstick{{}} & \qw  & \qw & \quad\ldots\quad & \qw  & \qw \\
\end{quantikz}
\end{adjustbox}
\caption{}
\label{fig:Ui}
    \end{subfigure}%
    \begin{subfigure}{0.4\textwidth}
        \centering
\begin{adjustbox}{width=0.6\textwidth}
\begin{quantikz}
\lstick{$1\ |0\rangle$} & \gate[style={fill=red!20}]{} \vqw{5} & \qw & \qw & \qw & \qw & \qw\\
\lstick{$2\ |0\rangle$} & \qw & \gate[style={fill=red!20}]{} \vqw{4} & \qw & \qw & \qw & \qw\\
\lstick{$3\ |0\rangle$} & \qw & \qw & \gate[style={fill=red!20}]{} \vqw{3} & \qw & \qw & \qw\\
\vdots \\
\lstick{$n\ |0\rangle$} & \qw & \qw & \qw &\quad \ldots\quad & \gate[style={fill=red!20}]{} \vqw{1} & \qw \\
\lstick{$\ |\underline{0}\rangle$} & \gate[style={fill=blue!20}]{}  & \gate[style={fill=blue!20}]{} & \gate[style={fill=blue!20}]{}&\quad \ldots\quad & \gate[style={fill=blue!20}]{} & \qw \\
\end{quantikz}
\end{adjustbox}
\caption{}
\label{fig:calU}
	 \end{subfigure}	
\caption{Circuit diagram for preparing the state $|D^{(s)}_{n,k}\rangle$
sequentially \eqref{sequentialspins}
(a) $U_i=\overset{\curvearrowleft}{\prod}_l I^{(i)}_l$, with $x=\max(0,2s(i - n - 1) + k)$ and $y=\min(2si-1,  k-1)$; 
(b) $\overset{\curvearrowleft}{\prod}_i U_i\, |\underline{0}\rangle|0\rangle^{\otimes n}$ }
\label{fig:Dicke}
\end{figure}
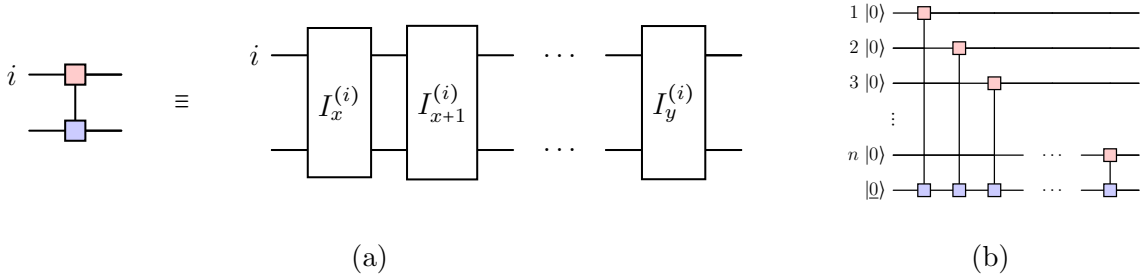

\subsection{QPE preparation}\label{sec:spinsDickeQPE}

\par We have seen in Sec. \ref{sec:spinsDickeSeq} that the sequential preparation of Dicke states $|D^{(s)}_{n,k}\rangle$ is deterministic, 
at the cost of circuit depth that grows linearly with $n$.
We consider here an alternative preparation method that has lower depth, but which is probabilistic. This approach is based on the quantum phase estimation algorithm \cite{Kitaev:1995qy, Nielsen:2019}, which was used in \cite{Wang:2021, Piroli:2024ckr} for preparing qubit ($s=1/2$) Dicke states. 

The two key steps in this approach are: 
\begin{enumerate}
\item constructing a suitable product state that can be expressed as a linear combination of Dicke states $|D^{(s)}_{n,k}\rangle$; and
\item exploiting the $U(1)$ symmetry of these states \eqref{U1spins} to select the desired one.
\end{enumerate}

For the first step, we observe that an $n$-fold tensor product of the 1-qudit state
\begin{equation}
    |\psi(s)\rangle=\frac{1}{2^s}\sum_{m=0}^{2s}\sqrt{\binom{2s}{m}}\, |m\rangle 
    \label{psidef}
\end{equation}
can be expressed as the following linear combination of Dicke states
\begin{equation}
    |\psi(s)\rangle^{\otimes n}=\frac{1}{2^{sn}}\sum_{k=0}^{2sn}\sqrt{\binom{2sn}{k}}\, |D^{(s)}_{n,k}\rangle \,.
    \label{no_boost_init_qpe}
\end{equation}
Indeed, from the definition \eqref{psidef}, it follows that
\begin{align}
    |\psi(s)\rangle^{\otimes n} &=\frac{1}{2^{sn}}
    \sum_{m_i = 0, 1, \ldots, 2s}
\sqrt{\binom{2s}{m_1} \binom{2s}{m_2} \cdots \binom{2s}{m_n}}\, 
|m_n \ldots m_2\, m_1\rangle \nonumber \\
&= \frac{1}{2^{sn}} \sum_{k=0}^{2sn} \sqrt{\binom{2sn}{k}}
\sum_{\substack{m_i = 0, 1, \ldots, 2s \\
m_1 + m_2 + \cdots + m_n = k}} 
\sqrt{
\frac{
\binom{2s}{m_1} \binom{2s}{m_2} \cdots \binom{2s}{m_n}
}{
\binom{2sn}{k}
}
}
\, |m_n \ldots m_2\, m_1\rangle \nonumber \\
&= \frac{1}{2^{sn}}\sum_{k=0}^{2sn}\sqrt{\binom{2sn}{k}}\, |D^{(s)}_{n,k}\rangle \,,    
\end{align}
where the last line follows from the identity \eqref{spin-s_Dicke_def}.

Borrowing a trick from \cite{Piroli:2024ckr}, let us now introduce into the 1-qudit state \eqref{psidef} a variational parameter $0 < p < 1$, which we will tune to boost the probability of preparing a Dicke state with a target value of $k$, see \eqref{qpe_spins-p} below. Hence, we instead make use of the 1-qudit state
\begin{equation}
    |\psi(s, p)\rangle=(1-p)^s\sum_{m=0}^{2s}
    \left(\sqrt{\frac{p}{1-p}}\right)^{m}
    \sqrt{\binom{2s}{m}}\, |m\rangle \,,
    \label{psip}
\end{equation}
which can be similarly shown to satisfy
\begin{equation}
    |\psi(s, p)\rangle^{\otimes n}=\sum_{k=0}^{2sn}
    (\sqrt{p})^k(\sqrt{1-p})^{(2sn-k)}
    \sqrt{\binom{2sn}{k}}\, |D^{(s)}_{n,k}\rangle \,.
      \label{init_qpe}
\end{equation}
The state $|\psi(s,p)\rangle$ \eqref{psip} can be prepared using a product of rotation 
gates \eqref{rotation gate} as follows
\begin{equation}
    |\psi(s, p)\rangle=R_{2s-1,2s}(\theta_{2s-1})\ldots R_{1,2}(\theta_1)R_{0,1}(\theta_0)|0\rangle \,,
    \label{qpe_creation by rotations}
\end{equation}
with the rotation angles
\begin{equation}
    \theta_i=2 \arccos\left( 
    \frac{(1-p)^s\sqrt{\binom{2s}{i} \left(\frac{p}{1-p}\right)^i}}
    {\prod_{j=0}^{i-1}\sin(\theta_j/2)} \right)
    \label{qpe_theta} \,,\qquad i = 0, 1, \ldots 2s-1 \,.
\end{equation}

For the second step, we define the $n$-qudit unitary operator
\begin{equation}
    U = \exp\left( 2 \pi i \mathbb{K}/2^\ell\right) 
     = \prod_{j=1}^n  \exp\left( 2 \pi i \mathbbm{k}_j/2^\ell\right) \,,
    \label{QPEUspins}
\end{equation}
where $\mathbb{K}$ and $\mathbbm{k}_j$ are defined in \eqref{Kdef}, and $\ell$ is still to be determined (see \eqref{ellvalue} below).
The Dicke state $|D^{(s)}_{n,k}\rangle$ is an eigenstate of this operator
\begin{equation}
    U\, |D^{(s)}_{n,k}\rangle 
    = e^{2 \pi i k/2^\ell}\, |D^{(s)}_{n,k}\rangle
\end{equation}
by virtue of \eqref{U1spins}. The QPE circuit (discussed in 
Sec. \ref{sec:spinslogdepth} below) uses the unitary operator \eqref{QPEUspins} to
project the product state \eqref{init_qpe} to the Dicke state $|D^{(s)}_{n,k}\rangle$ with a probability $P(k)$ given by
\begin{equation}
    P(k)=p^k (1-p)^{(2sn-k)} \binom{2sn}{k} \,,
\end{equation}
which is maximized for 
\begin{equation}
    p=\frac{k}{2sn} \,.
    \label{qpe_spins-p}
\end{equation}
The success probability of preparing $|D^{(s)}_{n,k}\rangle$ is therefore given by
\begin{equation}
    P(k) = \frac{(2 s n)!}{(2 s n)^{2 s n}} \frac{k^k}{k!}
    \frac{(2 s n - k)^{(2 s n - k)}}{(2 s n - k)!}
    \approx \sqrt{\frac{2sn}{2 \pi k (2sn - k)}} \,,
    \label{spinsprob}
\end{equation}
where we have used Stirling's approximation. In the worst case $k=s n$, the number of required repetitions is $1/P(k) = \mathcal{O}(\sqrt{s  n})$. For $k \ll s n$, fewer repetitions are needed, since then $1/P(k) = \mathcal{O}(\sqrt{k})$.

\subsubsection{Log depth}\label{sec:spinslogdepth}

The circuit diagram for the standard QPE algorithm is shown in Figure \ref{fig:log-spin-s}.
The bottom wire represents the $n$-qudit product state \eqref{init_qpe}.
There are $\ell$ qubit ancillas, where
$\ell$ is the minimum number of bits $k_i \in \{0, 1\}$ required to represent the maximum possible value of $k = \sum_{i=0}^{\ell-1} k_i 2^i$, namely,
\begin{equation}
\ell=\left\lceil \log_2\left(2sn + 1 \right) \right\rceil \,.
\label{ellvalue}
\end{equation}
The controlled unitaries are controlled versions of the unitary operator \eqref{QPEUspins}.
The state of the system just prior to the measurement is
\begin{equation}
   \sum_{k=0}^{2sn} \sqrt{P(k)} |D^{(s)}_{n,k}\rangle |k\rangle \,,
\end{equation}
where $P(k)$ is given by \eqref{spinsprob}.
The circuit therefore succeeds on measuring the ancilla qubits' base-10 value to be the $k$ of choice. 

\begin{figure}[H]
    \centering
  \begin{quantikz}[row sep=0.6cm, column sep=0.2cm]
   \lstick[4]{$\ell$} 
        &\wireoverride{n} 
        &\wireoverride{n}
         &\wireoverride{n}
          &\wireoverride{n}
    \lstick{$|0\rangle$}       & \gate{H}     & \ctrl{4}       & \qw             & \qw    & \qw     & \qw           & \gate[wires=4]{U_{\rm QFT}^{\dagger}} & \meter{k_0} & \qw \\
     &\wireoverride{n} 
        &\wireoverride{n}
         &\wireoverride{n}
          &\wireoverride{n}
    \lstick{$|0\rangle$}       & \gate{H}     & \qw            & \ctrl{3}          & \qw   & \qw     & \qw           &                                & \meter{k_1} & \qw \\
     &\wireoverride{n} 
        &\wireoverride{n}
         &\wireoverride{n}
          &\wireoverride{n}
    \lstick{$\vdots$}                                                                                          \\
     &\wireoverride{n} 
        &\wireoverride{n}
         &\wireoverride{n}
          &\wireoverride{n}
    \lstick{$|0\rangle$}   & \gate{H}     & \qw            & \qw             & \qw  & \qw      & \ctrl{1}      &                                & \meter{k_{l-1}} & \qw \\
     &\wireoverride{n} 
        &\wireoverride{n}
         &\wireoverride{n}
          &\wireoverride{n}
    \lstick{$|\psi(s,p)\rangle^{\otimes n}$}      & \qw          & \gate{U^{2^0}} & \gate{U^{2^1}}& \qw   & \ldots \ldots       & \gate{U^{2^{\ell-1}}} & \qw                             & \rstick{$|D^{(s)}_{n,k}\rangle$}
    \end{quantikz}    
    \caption{Circuit diagram for preparing the state $|D^{(s)}_{n,k}\rangle$
   in $\log$ depth using the standard QPE algorithm. All ancilla wires are qubits. The initial state of the bottom wire is \eqref{init_qpe},
   and $U$ is defined in \eqref{QPEUspins}.}
    \label{fig:log-spin-s}
\end{figure}
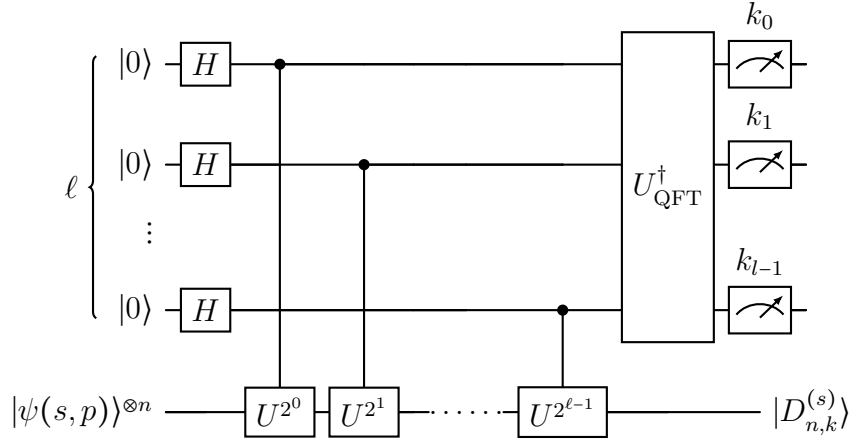

The circuit has $\ell$ controlled unitaries, each of which can be implemented in constant depth using mid-circuit measurement/feedforward and $n$ additional qubit ancillas, see Result 1 in \cite{Piroli:2024ckr}. Hence, the controlled unitaries can be implemented in depth $\mathcal{O}(\ell)$. The inverse quantum Fourier transform $U_{\rm QFT}^\dagger$ can also can be implemented in depth $\mathcal{O}(\ell)$. 
We therefore have the following:

\hypertarget{res: 2}{\noindent{\bf Result 2.}} The state $|D^{(s)}_{n,k}\rangle$ can be prepared probabilistically with at worst 
$\mathcal{O}(\sqrt{s  n})$ repetitions and with depth $\mathcal{O}(\ell) = \mathcal{O}(\log(s n))$, using $\mathcal{O}(\ell + n) = \mathcal{O}(\log(s n) + n)$ qubit ancillas. 

\noindent 
For $s=1/2$, this circuit reduces to Result 3--Proposition 1 in \cite{Piroli:2024ckr} with exact preparation.

\subsubsection{Constant depth}\label{sec:spinsconstantdepth}

Variations of the above circuit can prepare the state $|D^{(s)}_{n,k}\rangle$ in constant depth, at the expense of introducing additional and/or higher-dimensional ancillas.

\par The simplest such scheme, shown in Figure \ref{fig:simple-const-spin-s}, uses the Hadamard test with an auxiliary qudit (top wire) of dimension $d=2sn+1$, which is the number of possible values $k$-values. The gate $H_d$ is the generalized Hadamard gate (see e.g. \cite{Wang:2020})
\begin{equation}
    H_d\, |x\rangle=\frac{1}{\sqrt{d}}\sum_{y=0}^{d-1}  e^{2\pi i x y/d}\, |y\rangle \,,
    \label{Hadamard}
\end{equation}
and the controlled-$\mathcal{U}$ gate is defined as
\begin{equation}
    C\mathcal{U}\, |y\rangle |x\rangle  = 
    \left( \mathcal{U}(x)  |y\rangle \right) |x\rangle \,, \quad
   \mathcal{U}(x) =  \exp\left( 2 \pi i x \mathbb{K}/d \right)
     = \prod_{j=1}^n  \exp\left( 2 \pi i x \mathbbm{k}_j/d\right)
    \,.
    \label{calUdef}
\end{equation}
This gate can be implemented in constant depth using mid-circuit measurement/feedforward and $n$ additional ancilla qudits of dimension $d$, by a generalization of the proof of Result 1 in \cite{Piroli:2024ckr}, see also appendix A in \cite{Zi:2025dgw}. 
We therefore have the following:

\hypertarget{res: 3}{\noindent{\bf Result 3.}} The state $|D^{(s)}_{n,k}\rangle$ can be prepared probabilistically with at worst
$\mathcal{O}(\sqrt{s  n})$ repetitions and
with depth $\mathcal{O}(1)$, using $\mathcal{O}(n)$ ancillas of dimension $2sn+1$.

\noindent 
The scaling of an ancilla dimension with $n$ is an evident shortcoming of this approach.

\begin{figure}[H]
    \centering
    \begin{quantikz}[row sep=0.6cm, column sep=0.4cm]
    \lstick{$|0\rangle$} & \gate{H_d} & \gateO{*} \vqw{1} & \gate{H_d^{\dagger}} & \meter{k} \\
    \lstick{$|\psi(s,p)\rangle^{\otimes n}$} & \qw & \gate{\mathcal{U}} & \qw & \rstick{$|D^{(s)}_{n,k}\rangle$}
\end{quantikz}
    \caption{Circuit diagram for preparing the state $|D^{(s)}_{n,k}\rangle$
    in constant depth using the Hadamard test. The top wire is a qudit of dimension $d=2sn+1$. The initial state of the bottom wire is \eqref{init_qpe}, and $\mathcal{U}$ is defined in \eqref{calUdef}.}
    \label{fig:simple-const-spin-s}
\end{figure}
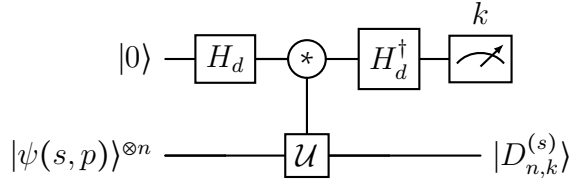

An alternative scheme, adapted from \cite{Piroli:2024ckr}, only requires ancillas of dimensions $2$ and $2s+1$. The initial state \eqref{init_qpe} can be written, with the help of the identity \eqref{spin-s_Dicke_def}, as
\begin{align}
    |\psi(s,p)\rangle^{\otimes n} 
    &= \sum_{k'=0}^{2sn} \sqrt{P(k')} |D^{(s)}_{n,k'}\rangle \nonumber \\
    &=\sum_{k'=0}^{2sn} \sum_{\substack{m_i = 0, 1, \ldots, 2s \\
m_1 + m_2 + \cdots + m_n = k'}} 
\; \alpha_{k',w} |m_n \ldots m_2\, m_1\rangle \,,   \qquad
    \alpha_{k',w} = \sqrt{P(k')} \sqrt{\frac{\binom{2s}{m_1} \binom{2s}{m_2} \cdots \binom{2s}{m_n}
}{\binom{2sn}{k'}}} \nonumber \\
    &=\sum_{k'=0}^{2sn}\, \sum_{w\in S(n,s,k')} \alpha_{k',w} |w\rangle \,, 
    \label{initialalt}
\end{align}
where $S(n,s,k')$ is the set of all permutations $w=m_n \ldots m_2\, m_1$ of $n$ integers, each of which is between 0 and $2s$, and which sum to $k'$; 
and $| w  \rangle =  |m_n \ldots m_2\, m_1\rangle$ is the computational basis state of $n$ qudits corresponding to the permutation $w$.
We fan out $\ell-1$ times the state  $|w\rangle$
(using $n(\ell-1)$ ancilla qudits of dimension $2s+1$, and corresponding qudit fan-out gates) to obtain the state
\begin{equation}
  \sum_{k'=0}^{2sn}\, \sum_{w\in S(n,s,k')} \alpha_{k',w} |w\rangle^{\otimes \ell}
\end{equation}
As shown in Figure \ref{fig:O(1)-spi}, using
$\ell$ qubit ancillas as controls, we then apply a product of controlled gates $V$ defined by
\begin{equation}
    V(x)=\left( |0\rangle\langle 0| \right) \otimes\mathbb{I}
    +\left(|1\rangle\langle 1| \right) 
    \otimes U(x) \,,
    \quad U(x) = \exp\left( 2\pi i(\mathbb{K}-k)/2^x \right) \,,
    \label{Vdef}
\end{equation}
where $\mathbb{K}$ is the $n$-qudit operator
defined in \eqref{Kdef}, and $k$ is the 
$k$-value of the target state $|D^{(s)}_{n,k}\rangle$. The states $|w\rangle$ are eigenstates of $ U(x)$ for any 
$w\in S(n,s,k')$,
\begin{equation}
    U(x)\, |w\rangle = e^{i \theta(x)}\, |w\rangle\,, \qquad 
    \theta(x) = 2\pi i(k'-k)/2^x \,.
\end{equation}
With the help of the identity (for integer values of $k$ and $k'$)
\begin{equation}
    \prod_{x=1}^\ell \left(1+e^{i \theta(x)} \right) = 2^\ell \delta_{k',k} \,,
    \label{identitytheta}
\end{equation}
one can see that the state of the system just prior to measurement is 
\begin{equation}
    \sqrt{P(k)}\, |\underline{0}\rangle^{\otimes n(\ell-1)}|D^{(s)}_{n,k}\rangle |0\rangle^{\otimes \ell}\, + \ldots \,.
\end{equation}
The circuit therefore succeeds on measuring all $\ell$ qubit ancillas to be zero. 

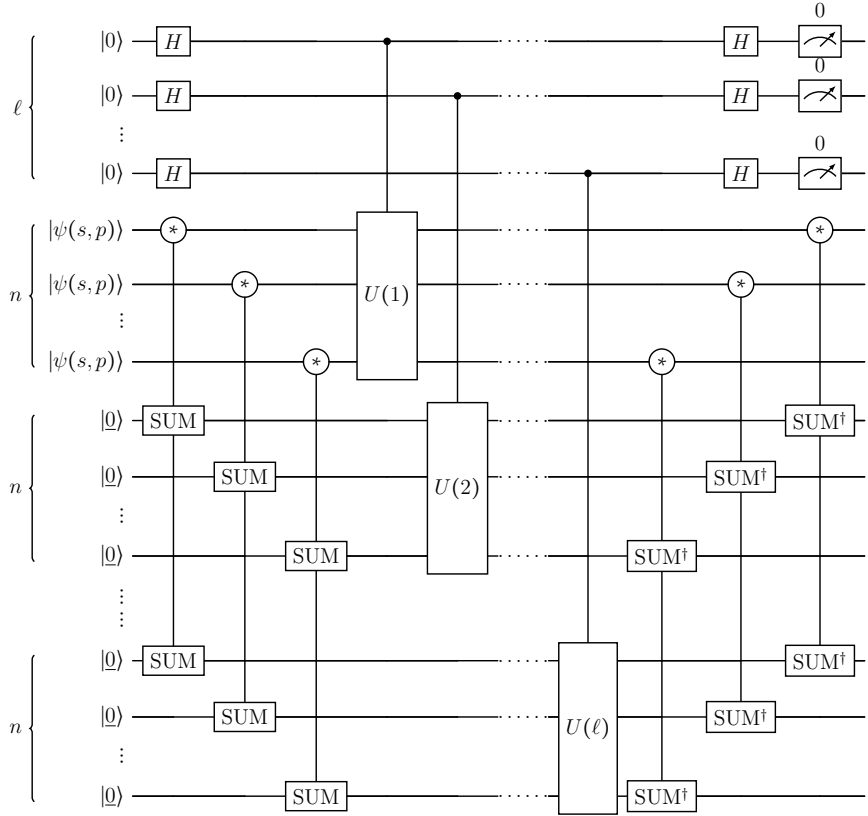
\begin{figure}[H]
    \centering
\begin{adjustbox}{width=0.7\textwidth}
  \begin{quantikz}[row sep=0.45cm, column sep=0.2cm]
   \lstick[4]{$\ell$} 
        &\wireoverride{n} 
        &\wireoverride{n}
         &\wireoverride{n}
          &\wireoverride{n}
           &\wireoverride{n}
          &\wireoverride{n}
          &\wireoverride{n}
           &\wireoverride{n}
          &\wireoverride{n}
    \lstick{$|0\rangle$}       & \gate{H}  & \qw  & \qw     & \ctrl{4}       & \qw             & \ldots\ldots    & \qw     & \qw           & \gate{H} & \meter{0} & \qw \\
    &\wireoverride{n} 
        &\wireoverride{n}
         &\wireoverride{n}
          &\wireoverride{n}
           &\wireoverride{n}
          &\wireoverride{n}
          &\wireoverride{n}
           &\wireoverride{n}
          &\wireoverride{n}
    \lstick{$|0\rangle$}       & \gate{H}  & \qw & \qw   & \qw            & \ctrl{10}          & \ldots\ldots   & \qw     & \qw  &\gate{H}                               & \meter{0} & \qw \\
    &\wireoverride{n} 
        &\wireoverride{n}
         &\wireoverride{n}
          &\wireoverride{n}
           &\wireoverride{n}
          &\wireoverride{n}
          &\wireoverride{n}
           &\wireoverride{n}
          &\wireoverride{n}
    \lstick{$\vdots$}                                                                                          \\
    &\wireoverride{n} 
        &\wireoverride{n}
         &\wireoverride{n}
          &\wireoverride{n}
           &\wireoverride{n}
          &\wireoverride{n}
          &\wireoverride{n}
           &\wireoverride{n}
          &\wireoverride{n}
    \lstick{$|0\rangle$}   & \gate{H}     & \qw            & \qw             & \qw  & \qw     & \ldots \ldots     & \ctrl{11}  & \qw   &\gate{H}                                & \meter{0} & \qw \\
      \lstick[4]{$n$} 
        &\wireoverride{n} 
        &\wireoverride{n}
         &\wireoverride{n}
          &\wireoverride{n}
           &\wireoverride{n}
          &\wireoverride{n}
          &\wireoverride{n}
           &\wireoverride{n}
          &\wireoverride{n}
    \lstick{$|\psi(s,p)\rangle$} &  \gateO{*} \vqw{4}   & \qw          & \qw &  \gate[wires=4]{U(1)}  & \qw   & \ldots \ldots       & \qw   & \qw  & \qw &  \gateO{*} \vqw{4} & \qw  \\ 
     &\wireoverride{n} 
        &\wireoverride{n}
         &\wireoverride{n}
          &\wireoverride{n}
           &\wireoverride{n}
          &\wireoverride{n}
          &\wireoverride{n}
           &\wireoverride{n}
          &\wireoverride{n}
    \lstick{$|\psi(s,p)\rangle$}      & \qw & \gateO{*} \vqw{4}         & \qw   & \qw  & \qw   & \ldots \ldots       & \qw   & \qw & \gateO{*} \vqw{4}  & \qw & \qw  \\ 
     &\wireoverride{n} 
        &\wireoverride{n}
         &\wireoverride{n}
          &\wireoverride{n}
           &\wireoverride{n}
          &\wireoverride{n}
          &\wireoverride{n}
           &\wireoverride{n}
          &\wireoverride{n}
    \lstick{$\vdots$}   \\
        &\wireoverride{n} 
        &\wireoverride{n}
         &\wireoverride{n}
          &\wireoverride{n}
           &\wireoverride{n}
          &\wireoverride{n}
          &\wireoverride{n}
           &\wireoverride{n}
          &\wireoverride{n}
    \lstick{$|\psi(s,p)\rangle$}      & \qw          & \qw & \gateO{*} \vqw{4} & \qw  & \qw   & \ldots \ldots       & \qw  & \gateO{*} \vqw{4} & \qw  & \qw & \qw  \\  
    \lstick[4]{$n$} 
     &\wireoverride{n} 
        &\wireoverride{n}
         &\wireoverride{n}
          &\wireoverride{n}
           &\wireoverride{n}
          &\wireoverride{n}
          &\wireoverride{n}
           &\wireoverride{n}
          &\wireoverride{n}
    \lstick{$|\underline{0}\rangle$} & \gate{\rm SUM} \vqw{6}    & \qw          & \qw   & \qw  &  \gate[wires=4]{U(2)}  & \ldots \ldots       & \qw   & \qw  & \qw   & \gate{\rm SUM^{\dagger}} \vqw{6} & \qw  \\ 
     &\wireoverride{n} 
        &\wireoverride{n}
         &\wireoverride{n}
          &\wireoverride{n}
           &\wireoverride{n}
          &\wireoverride{n}
          &\wireoverride{n}
           &\wireoverride{n}
          &\wireoverride{n}
    \lstick{$|\underline{0}\rangle$}      & \qw    & \gate{\rm SUM} \vqw{6}       & \qw   & \qw  & \qw   & \ldots \ldots       & \qw   & \qw & \gate{\rm SUM^{\dagger}} \vqw{6} & \qw  & \qw  \\ 
     &\wireoverride{n} 
        &\wireoverride{n}
         &\wireoverride{n}
          &\wireoverride{n}
           &\wireoverride{n}
          &\wireoverride{n}
          &\wireoverride{n}
           &\wireoverride{n}
          &\wireoverride{n}
    \lstick{$\vdots$}   \\
     &\wireoverride{n} 
        &\wireoverride{n}
         &\wireoverride{n}
          &\wireoverride{n}
           &\wireoverride{n}
          &\wireoverride{n}
          &\wireoverride{n}
           &\wireoverride{n}
          &\wireoverride{n}
    \lstick{$|\underline{0}\rangle$}      & \qw          & \qw  & \gate{\rm SUM} \vqw{6}  & \qw  & \qw   & \ldots \ldots       & \qw & \gate{\rm SUM^{\dagger}} \vqw{6}   & \qw  & \qw  & \qw \\ 
     &\wireoverride{n} 
        &\wireoverride{n}
         &\wireoverride{n}
          &\wireoverride{n}
           &\wireoverride{n}
          &\wireoverride{n}
          &\wireoverride{n}
           &\wireoverride{n}
          &\wireoverride{n}
    \lstick{$\vdots$}      \\ 
     &\wireoverride{n} 
        &\wireoverride{n}
         &\wireoverride{n}
          &\wireoverride{n}
           &\wireoverride{n}
          &\wireoverride{n}
          &\wireoverride{n}
           &\wireoverride{n}
          &\wireoverride{n}
    \lstick{$\vdots$}      \\
     \lstick[4]{$n$} 
     &\wireoverride{n} 
        &\wireoverride{n}
         &\wireoverride{n}
          &\wireoverride{n}
           &\wireoverride{n}
          &\wireoverride{n}
          &\wireoverride{n}
           &\wireoverride{n}
          &\wireoverride{n}
    \lstick{$|\underline{0}\rangle$}  & \gate{\rm SUM}   & \qw          & \qw   & \qw  & \qw   & \ldots \ldots       &   \gate[wires=4]{U(\ell)}    & \qw & \qw &  \gate{\rm SUM^{\dagger}} & \qw \\ 
     &\wireoverride{n} 
        &\wireoverride{n}
         &\wireoverride{n}
          &\wireoverride{n}
           &\wireoverride{n}
          &\wireoverride{n}
          &\wireoverride{n}
           &\wireoverride{n}
          &\wireoverride{n}
    \lstick{$|\underline{0}\rangle$}      & \qw  & \gate{\rm SUM}         & \qw   & \qw  & \qw   & \ldots \ldots         & \qw  & \qw & \gate{\rm SUM^{\dagger}} & \qw & \qw   \\ 
    &\wireoverride{n} 
        &\wireoverride{n}
         &\wireoverride{n}
          &\wireoverride{n}
           &\wireoverride{n}
          &\wireoverride{n}
          &\wireoverride{n}
           &\wireoverride{n}
          &\wireoverride{n}
    \lstick{$\vdots$}   \\
    &\wireoverride{n} 
        &\wireoverride{n}
         &\wireoverride{n}
          &\wireoverride{n}
           &\wireoverride{n}
          &\wireoverride{n}
          &\wireoverride{n}
           &\wireoverride{n}
          &\wireoverride{n}
    \lstick{$|\underline{0}\rangle$}      & \qw          & \qw  & \gate{\rm SUM}  & \qw  & \qw   & \ldots \ldots      & \qw & \gate{\rm SUM^{\dagger}}  & \qw & \qw & \qw 
\end{quantikz} 
\end{adjustbox}
    \caption{Circuit diagram for preparing the state $|D^{(s)}_{n,k}\rangle$, which can be implemented
    in constant depth. The top $\ell$ wires are qubits, while all other wires are qudits of dimension $2s+1$. The state $|\psi(s,p)\rangle$ is given by \eqref{psip}, and
    $U(x)$ is defined in \eqref{Vdef}.} 
    \label{fig:O(1)-spi}
\end{figure}

The qudit fan-out gates (represented in Figure \ref{fig:O(1)-spi} by ${\rm SUM}$ gates) can be implemented in constant depth (see appendix A in \cite{Zi:2025dgw}), and likewise for the $V(x)$ gates (see Result 1 in \cite{Piroli:2024ckr}).
We therefore have the following:

\hypertarget{res: 4}{\noindent{\bf Result 4.}} The state $|D^{(s)}_{n,k}\rangle$ can be prepared probabilistically with at worst
$\mathcal{O}(\sqrt{s  n})$ repetitions and
with depth $\mathcal{O}(1)$, using $l$ qubit ancillas, and $n(\ell-1)$ qudit ancillas of dimension $2s+1$, where $\ell$ is given by \eqref{ellvalue}.

\noindent
For $s=1/2$, this circuit reduces to Result 3--Proposition 2 in \cite{Piroli:2024ckr} with exact preparation.

\section{$SU(d)$ Dicke states $|D^{n}(\vec k)\rangle$}\label{sec:sudDicke}

Just as ordinary $SU(2)$ Dicke states $|D^{(1/2)}_{n,k}\rangle$
are specified by a fixed number $k$ of $|1\rangle$’s (and therefore $n-k$ $|0\rangle$'s), $SU(d)$ Dicke states are characterized by a fixed vector $\vec k$ of occupation numbers for each of $d$ levels -- that is, a specified number of $d$-level qudits occupying each level. Given $\vec k$, an $SU(d)$ Dicke state is constructed as a uniform superposition over all computational basis states that match the specified occupation numbers.

\par More explicitly, let $\vec k=(k_0,k_1,\ldots,k_{d-1})$ be a vector of $d$ integers, each of which is between $0$ and $n$, and which sum to $n$ (that is, 
$k_i \in \{0, 1, \ldots, n\}$, and $\sum_{i=0}^{d-1} k_i=n$). The corresponding 
$n$-qudit $SU(d)$ Dicke state $ |D^n(\vec k)\rangle$ is defined by
\begin{equation}
    |D^n(\vec k)\rangle=\frac{1}{\sqrt{\binom{n}{\vec k}}}\sum_{w\in\mathfrak{S}_{M(\vec k)}}|w\rangle \,,
    \label{SU(d) Dicke}
\end{equation}
where $\mathfrak{S}_{M(\vec k)}$ is the set of all permutations of the multiset $M(\vec k)$ 
\begin{equation}
M(\vec k)	=\{ \underbrace{0, \ldots, 0}_{k_{0}}, \underbrace{1, 
\ldots, 1}_{k_{1}}, \ldots, \underbrace{d-1, \ldots, d-1}_{k_{d-1}}\} \,,
\label{multiset}
\end{equation}
such that $k_i$ is the multiplicity of $i$ in $M(\vec k)$, and the cardinality of $M(\vec k)$ is $n$; 
and $| w  \rangle$ is the computational basis state of $n$ qudits corresponding 
to the permutation $w$. Furthermore, ${n \choose \vec k}$ denotes the multinomial
\begin{equation}
{n \choose \vec k} = {n \choose k_{0}, k_{1}, \ldots, k_{d-1}}	= 
\frac{n!}{\prod_{i=0}^{d-1}k_{i}!} \,.
\label{multinomial}
\end{equation}
An example with $\vec k=(1,1,1)$, so that $d=3$ (qutrits) and $n=3$, is 
given by \eqref{sudexample}.

\par Approaches for preparing $SU(d)$ Dicke states were presented in 
\cite{Nepomechie:2023lge, Liu:2024taj}. In this section we present several additional methods of preparing these states: a sequential deterministic approach in Sec. \ref{sec:sudDickeSeq}, and a probabilistic approach based on QPE in Sec. \ref{sec:sudDickeQPE}.

\subsection{Sequential preparation}\label{sec:sudDickeSeq}

\par An exact canonical MPS representation for $|D^n(\vec k)\rangle$
was derived in \cite{Raveh:2024sku}. The basis for the MPS ancilla consists of \emph{level sets} $\mathcal{A}^i(\vec{k})$ defined by
\begin{equation}
    \mathcal{A}^i(\vec{k}) = \big\{ \vec{a} = (a_0, a_1, \ldots, a_{d-1})\, \big\vert\quad 0 \le a_j \le k_j \,, \quad \sum_{j=0}^{d-1} a_j = i \big\} \,, \qquad i = 0, 1, \ldots, n \,.
    \label{setA}
\end{equation}
The elements $\vec{a} \in \mathcal{A}^i(\vec{k})$ are labeled (indexed) by consecutive integers $J^i(\vec{a}) = 0, 1, \ldots, \mathcal{D}^i(\vec{k})-1 $, where $\mathcal{D}^i(\vec{k})=|\mathcal{A}^i(\vec{k})|$ is the cardinality of $\mathcal{A}^i(\vec{k})$. Here we order each level set in \emph{reverse lexicographic order}, which is central for our construction, see appendix \ref{sec:proof}. Hence, for $\vec{x}, \vec{y} \in  \mathcal{A}^i(\vec{k})$, we assign their labels such that
\begin{equation}
J^i(\vec{x}) < J^i(\vec{y}) \quad \iff \quad \vec{x} >_{\text{lex}} \vec{y},
\end{equation}
where lexicographic order compares vectors from left to right. Thus, $\vec{x}=(x_0, x_1, \ldots, x_{d-1}) >_{\text{lex}}\ \vec{y}=(y_0, y_1, \ldots, y_{d-1})$ if $x_j > y_j$ for the first $j$ where $x_j \ne y_j$. 

For example, for $\vec{k}=(1,1,1)$ so that $d=3$, the level set $\mathcal{A}^1(\vec{k})$ is given by
\begin{equation}
    \mathcal{A}^1(\vec{k})  =  \left\{ (1,0,0)\,,  (0, 1, 0)\,,  (0, 0, 1) \right\}
     = \{ \hat{0}, \hat{1}, \hat{2} \} \,,
    \label{exampleA1}
\end{equation}
where $\hat{m}$ is a $d$-dimensional unit vector in the $m$-th direction; 
that is, it has components $(\hat{m})_i =\delta_{m,i}$, with $m=0, 1, \dots, d-1$.
The elements of $\mathcal{A}^1(\vec{k})$ are ordered in \eqref{exampleA1} in 
reverse lexicographic order, and have labels $J^1(\hat{m})=m$ for $m= 0, 1, 2$.

A canonical MPS with minimum bond dimension $\chi=\mathcal{D}^{\lfloor n/2 \rfloor}(\vec{k})$ is given by \cite{Raveh:2024sku}
\begin{equation}
|D^n(\vec{k})\rangle=\sum_{\vec m}\langle \underline{0}|A_n^{m_n}\dots A_2^{m_2}A_1^{m_1}|\underline 0\rangle\,|\vec m\rangle\,,
\label{quditDickeMPSspins}
\end{equation} 
where $A_i^m$ are $\chi \times \chi$  matrices with elements 
\begin{equation}
    \langle \underline{J^i(\vec{a'})}| A_i^m | \underline{J^{i-1}(\vec{a})} \rangle = \gamma^{(i)}_{J^{i-1}(\vec{a}),m}\, \delta_{\vec{a'},\vec{a}+\hat{m}} \,,
    \qquad \gamma^{(i)}_{J^{i-1}(\vec{a}),m} = \sqrt{\frac{\binom{n-i}{\vec{k}-\vec{a}-\hat{m}}} {\binom{n-i+1}{\vec{k}-\vec{a}}}} \,,
    \label{Amatelemsqudit}
\end{equation}
where $\vec{a} \in \mathcal{A}^{i-1}(\vec{k})$ and $\vec{a'} \in \mathcal{A}^{i}(\vec{k})$. The coefficient $\gamma^{(i)}_{J^{i-1}(\vec{a}),m}$ is zero unless $\vec{a} + \hat{m} \in \mathcal{A}^i(\vec{k})$.

The fact that the MPS is canonical implies \cite{Schon:2005} that we can define a two-qudit unitary operator $U_i$ acting on an ancilla qudit $|\underline{J^{i-1}(\vec{a})}\rangle$ (for $\vec{a} \in \mathcal{A}^{i-1}(\vec{k})$)
and a ``system'' qudit at site $i \in \{1, 2, \ldots, n\}$ that performs the mapping
\begin{equation}
U_i\, |\underline{J^{i-1}(\vec{a})}\rangle\, |0\rangle_i 
= \sum_{m=0}^{d-1} \left( A^m_i\, |\underline{J^{i-1}(\vec{a})}\rangle \right)
|m\rangle_i 
= \sum_{m=0}^{d-1} \gamma^{(i)}_{J^{i-1}(\vec{a}),m}\,  
|\underline{J^{i}(\vec{a} + \hat{m})}\rangle\, |m\rangle_i  \,.
\label{u-gate-su(d)}
\end{equation}
It follows that the state $|D^{n}(\vec k)\rangle$ can be prepared sequentially as follows
\begin{equation}
    |\underline{0}\rangle|D^{n}(\vec k)\rangle=U_n\ldots U_2\, U_1|\underline{0}\rangle |0\rangle^{\otimes n} \,.
    \label{seqsud}
\end{equation}. 

In order to formulate a circuit implementation of the sequential preparation \eqref{seqsud}, it is necessary to devise a circuit for the unitary $U_i$ 
\eqref{u-gate-su(d)}. Proceeding as in Sec. \ref{sec:spinsDickeSeq},
we decompose $U_i$ into an ordered product of simpler operators 
${\rm I}^{(i)}_{J^{i-1}(\vec{a})}$
\begin{equation}
    U_i=\overset{\curvearrowleft}{\prod_{\vec{a} \in \mathcal{A}^{i-1}(\vec{k})}} 
    {\rm I}^{(i)}_{J^{i-1}(\vec{a})} \,,
 \label{sudU}
\end{equation}
where the product goes from right to left with increasing label $J^{i-1}(\vec{a})$, 
such that 
\begin{numcases}
    {{\rm I}^{(i)}_p |\underline{J^{i-1}(\vec{a}}\rangle|0\rangle_i =}
       \sum_{m=0}^{d-1} \gamma^{(i)}_{J^{i-1}(\vec{a}),m}
       |\underline{J^{i}(\vec{a} + \hat{m})}\rangle|m\rangle_i 
       & if \hspace{.2cm } $p=J^{i-1}(\vec{a})$ \label{eq:sudIa} \\ 
      |\underline{J^{i-1}(\vec{a}}\rangle|0\rangle_i
       & if \hspace{.2cm } $p < J^{i-1}(\vec{a})$ \label{eq:sudIb}
\end{numcases}
and
\begin{equation}
{\rm I}^{(i)}_p \left( {\rm I}^{(i)}_{J^{i-1}(\vec{a})} |\underline{J^{i-1}(\vec{a}}\rangle
|0\rangle_i \right) 
= \left( {\rm I}^{(i)}_{J^{i-1}(\vec{a})} |\underline{J^{i-1}(\vec{a}}\rangle
|0\rangle_i \right)
\qquad\qquad\qquad  \text{if }p > J^{i-1}(\vec{a}) \,. 
    \label{eq:sudIc}
\end{equation}

The operators ${\rm I}^{(i)}_{J^{i-1}(\vec{a})}$ can be implemented by the quantum circuit whose circuit diagram is shown in Figure \ref{fig:I_SU(d)-expanded}. The top wire is a ``system'' qudit of dimension $d$, and the middle wire is the MPS ancilla qudit of dimension $\chi$. In contrast with the corresponding $SU(2)$ spin-$s$ circuit in Figure \ref{fig:I_spin-s}, there is an additional (bottom) wire representing a qubit ancilla. The level sets $\mathcal{A}^{i}(\vec{k})$, the corresponding labels $J^{i}(\vec{a})$ and their inverses 
$\left(J^{i}\right)^{-1} \in \mathcal{A}^i(\vec{k})$
are computed classically.

In Figure \ref{fig:I_SU(d)-expanded},
for the case that both $x=0$ and $j=J^{i-1}(\vec{a})$, corresponding to the first condition \eqref{eq:sudIa}, the state of the qubit ancilla is flipped to $|1\rangle$ by the double-controlled NOT; the controlled rotation gates \eqref{rotation gate} with angles
\begin{equation}
     \theta_m= 2 \arccos\left( \frac{\gamma^{(i)}_{J^{i-1}(\vec{a}),m}}
     {\prod_{p=0}^{m-1}\sin(\theta_p/2)} \right) \,,
     \qquad m = 0, 1, \ldots, d-2 \,,
     \label{thetasud}
\end{equation}
generate the coefficients $\gamma^{(i)}_{J^{i-1}(\vec{a}),m}$ in \eqref{eq:sudIa}.
After the rotations, the state of the MPS ancilla qudit is mapped to 
$|\underline{J^{i}(\vec{a} + \hat{m})}\rangle$ for $m=0, 1, \ldots, d-1$ by a series of double-controlled-$X^{i,j}$ gates, where the 1-qudit gate $X^{i,j}$ is defined as
\begin{equation}
    X^{i,j}\, |i\rangle = |j\rangle \,, \qquad 
    X^{i,j}\, |j\rangle = |i\rangle \,.
\end{equation}
Finally, the qubit ancilla is reset to $|0\rangle$ by a series of double-controlled 
NOTs.

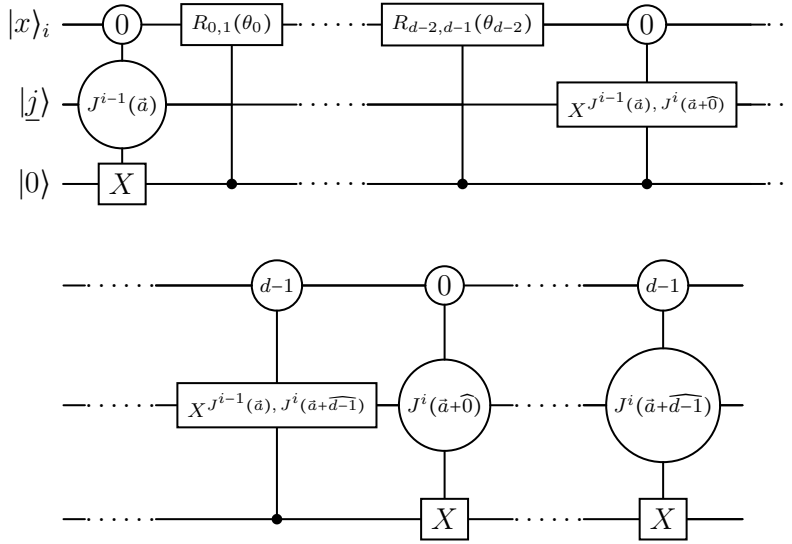
\begin{figure}[H]
  \centering
  \scalebox{0.95}{
  \begin{tabular}{c}
  \begin{quantikz}[row sep=0.2cm, column sep=0.2cm]
    \lstick{$|x\rangle_i$} 
    & \gateO{0} \vqw{1} 
    & \gate{\scriptstyle R_{0,1}(\theta_0)} \vqw{2} 
    & \ldots \ldots
    & \gate{\scriptstyle R_{d-2,d-1}(\theta_{d-2})} \vqw{1} 
    & \gateO{0} \vqw{1} 
    & \qw & \ldots \\
    
    \lstick{$|\underline{j}\rangle$} 
    & \gateO{\scriptstyle J^{i{-}1}(\vec{a})} \vqw{1} 
    & \qw 
    & \ldots \ldots
    & \qw \vqw{1}
    & \gate{\scriptstyle X^{J^{i{-}1}(\vec{a}),\,J^{i}(\vec{a}+\widehat{0})}} \vqw{1}
    & \qw & \ldots\\

    \lstick{$|0\rangle$} 
    & \gate{X}
    & \control{}
    & \ldots \ldots
    & \control{}
    & \control{}
    & \qw & \ldots
  \end{quantikz}
  \\
  \\
  \begin{quantikz}[row sep=0.5cm, column sep=0.3cm]
    %\lstick{} 
    & \ldots \ldots
    & \gateO{\scriptstyle d{-}1} \vqw{1} 
    & \gateO{0} \vqw{1} 
    & \ldots \ldots
    & \gateO{\scriptstyle d{-}1} \vqw{1} 
    & \qw \\

   % \lstick{} 
    
    & \ldots \ldots
    & \gate{\scriptstyle X^{J^{i{-}1}(\vec{a}),\,J^{i}(\vec{a}+\widehat{d-1})}} \vqw{1}
    & \gateO{\scriptstyle J^{i}(\vec{a}+\widehat{0})} \vqw{1} 
    & \ldots \ldots
    & \gateO{\scriptstyle J^{i}(\vec{a}+\widehat{d-1})} \vqw{1} 
    & \qw \\

   % \lstick{} 
    
    & \ldots \ldots
    & \control{}
    & \gate{X}
    &\ldots \ldots
    & \gate{X}
    & \qw
  \end{quantikz}
  \end{tabular}
  }
  \caption{Circuit diagram for $I^{(i)}_{J^{i-1}(\vec{a})}$}
  \label{fig:I_SU(d)-expanded}
\end{figure}

The second condition \eqref{eq:sudIb} requires $I^{(i)}_{J^{i-1}(\vec{a})} 
|\underline{j}\rangle|0\rangle_i = |\underline{j}\rangle|0\rangle_i$ for 
$j > J^{i-1}(\vec{a})$. In order for the $I^{(i)}_{J^{i-1}(\vec{a})}$ circuit in Figure \ref{fig:I_SU(d)-expanded} to satisfy this condition, it is necessary (to avoid triggering the double-controlled-NOT with the control value $J^{i}(\vec{a} + \hat{0})$
near the end of the circuit) that 
\begin{equation}
j > J^{i-1}(\vec{a}) \Rightarrow j \ne J^{i}(\vec{a} + \hat{0}) \,,
    \label{needed}
\end{equation}
which is guaranteed by our labeling of the level sets in reverse lexographic order. Indeed, we exploit this ordering to
show in appendix \ref{sec:proof} that
\begin{equation}
    J^{i-1}(\vec{a}) \ge  J^{i}(\vec{a}+\hat{0}) 
    \label{proved}
\end{equation}
for all $\vec{a} \in \mathcal{A}^{i-1}(\vec{k})$, from which \eqref{needed} follows.

Finally, it is straightforward to check that the third condition \eqref{eq:sudIc},
which ensures that these gates do not interfere,
is also satisfied by the circuit in Figure \ref{fig:I_SU(d)-expanded}.
The full circuit for the sequential preparation \eqref{seqsud},
including all $U_i$ operators, has the same structure as in Figure \ref{fig:calU}.

We see from \eqref{sudU} that each $U_i$ is made of 
$\mathcal{D}^{i-1}(\vec{k})$-many ${\rm I}^{(i)}$-operators; and we see from Figure \ref{fig:I_SU(d)-expanded} that each ${\rm I}^{(i)}$-operator is made of $3d$ gates. Hence, the total circuit depth is $3d\sum_{i=1}^n \mathcal{D}^{i-1}(\vec{k})$. Although a closed-form expression for $\mathcal{D}^{i-1}(\vec{k})$ is not known, it has the ``stars and stripes'' bound 
$\mathcal{D}^{i-1}(\vec{k}) \le \binom{i+d-2}{i-1}$, which leads to an approximate circuit depth $\mathcal{O}((n/d)^d)$. A similar result is obtained by considering the worst case $\vec{k}=(\frac{n}{d},\frac{n}{d},\ldots,\frac{n}{d})$, for which $\sum_{i=1}^n \mathcal{D}^{i-1}(\vec{k}) \approx (\frac{n}{d} +1)^d$. 
We therefore have the following:

\hypertarget{res: 5}{\noindent{\bf Result 5.}} The state $|D^{n}(\vec k)\rangle$ can be prepared deterministically with a worst-case approximate depth $\mathcal{O}((n/d)^d)$, 
using one qubit ancilla, and one
ancilla of dimension $\chi=\mathcal{D}^{\lfloor n/2 \rfloor}(\vec{k}) \le 
\binom{{\lfloor n/2 \rfloor}+d-1}{{\lfloor n/2 \rfloor}}$.

\noindent
The circuit depth is significantly smaller for typical $\vec{k}$-values. For example, for $\vec{k}$ of the form
\begin{equation}
    \vec{k} = (n - r x, \overbrace{x, \ldots,x}^r, 
    \overbrace{0, \ldots,0}^{d-r-1}) \,,
    \qquad 0 < r\le d-1\,, \qquad 0 < x \le \frac{n}{r+1}\,,
    \label{better}
\end{equation}
or any permutation thereof, we find that $\chi \le (x+1)^r$, implying an approximate depth $\mathcal{O}(x^r n)$.

\subsection{QPE preparation}\label{sec:sudDickeQPE}

We now consider a probabilistic approach of preparing the Dicke states 
$|D^n(\vec k)\rangle$ that is based on the quantum phase estimation algorithm \cite{Kitaev:1995qy, Nielsen:2019}. As in the case of $SU(2)$ spin-$s$ Dicke states discussed in Sec. \ref{sec:spinsDickeQPE}, there are two key steps:
\begin{enumerate}
    \item constructing a suitable product state that can be expressed as a linear combination of Dicke states $|D^n(\vec k)\rangle$; and
    \item exploiting the $U(1)^{\otimes (d-1)}$ symmetry of these states (see \eqref{U1sud} below) to select the desired one.
\end{enumerate}

\par For the first step, we observe that the $n$-fold tensor product of the 1-qudit state
\begin{equation}
    |\psi(d)\rangle = \frac{1}{\sqrt{d}} \sum_{m=0}^{d-1}|m\rangle 
    \label{sudqpestate0}
\end{equation}
can be expressed as the following linear combination of Dicke states
\begin{equation}
    |\psi(d)\rangle^{\otimes n} = \frac{1}{d^{\frac{n}{2}}}
    \sum_{\substack{k_i =0, 1, \ldots, n \\
k_0 + k_1 + \ldots + k_{d-1} = n}} \sqrt{\binom{n}{\vec k}}\, |D^n(\vec k)\rangle \,.
\end{equation}
Indeed, 
\begin{align}
   |\psi(d)\rangle^{\otimes n} &= \frac{1}{d^{\frac{n}{2}}} 
   \sum_{m_j = 0, \ldots, d-1}\, |m_{n} \ldots m_{1} \rangle \nonumber \\
   &= \frac{1}{d^{\frac{n}{2}}} \sum_{\substack{k_i =0, \ldots, n \\
k_0 + \ldots + k_{d-1} = n}}\, 
\sum_{w\in\mathfrak{S}_{M(\vec k)}}|w\rangle \nonumber \\
&= \frac{1}{d^{\frac{n}{2}}}
    \sum_{\substack{k_i =0, \ldots, n \\
k_0 + \ldots + k_{d-1} = n}} \sqrt{\binom{n}{\vec k}}\, 
|D^n(\vec k)\rangle \,.
\label{sudqpeproof}
\end{align}
In passing to the second line of \eqref{sudqpeproof}, we used the fact 
\begin{align}
    |w\rangle &= |m_{n} \ldots m_{1} \rangle \quad \text{with} \quad
    m_j \in \{0, \ldots, d-1 \} \quad \text{for all} \quad j=1, \ldots, n \\ 
    & \iff w \in\mathfrak{S}_{M(\vec k)} \quad \text{with} \quad 
    k_i \in \{0, \ldots, n \} \quad \text{for all} \quad i=0, \ldots, d-1
    \quad \text{and} \quad \sum_{i=0}^{d-1} k_i = n \,, \nonumber
\end{align}
where the multiset $M(\vec k)$ is defined in \eqref{multiset}; and the third line of \eqref{sudqpeproof} follows from the definition \eqref{SU(d) Dicke} of the Dicke state $|D^n(\vec k)\rangle$.

To boost the probability of preparing a Dicke state with a target $\vec{k}$-value, we introduce $d$ variational parameters $0 < \xi_i < 1$ into the 1-qudit state \eqref{sudqpestate0}
\begin{equation}
    |\psi(d, \vec{\xi})\rangle = \frac{1}{\sqrt{\vec{\xi} \cdot \vec{\xi}}} \sum_{m=0}^{d-1}\xi_m\, |m\rangle \,,
    \label{sudqpestate}
\end{equation}
which can be similarly shown to satisfy
\begin{equation}
    |\psi(d, \vec{\xi})\rangle^{\otimes n} = 
   \frac{1}{(\vec{\xi} \cdot \vec{\xi})^{\frac{n}{2}}}
    \sum_{\substack{k_i =0, 1, \ldots, n \\
k_0 + k_1 + \ldots + k_{d-1} = n}} \left(\prod_{i=0}^{d-1}\xi_i^{k_i}\right) 
\sqrt{\binom{n}{\vec k}}\, |D^n(\vec k)\rangle \,.
\label{qpesudprod}
\end{equation}
The 1-qudit state \eqref{sudqpestate} can be prepared using a product of rotation gates similarly to  \eqref{qpe_creation by rotations}.

For the second step, we define the $d-1$ Hermitian and mutually-commuting operators 
$\mathbb{K}^{(1)}, \ldots, \mathbb{K}^{(d-1)}$ as follows
\begin{equation}
    \mathbb{K}^{(i)} = \sum_{j=1}^n \mathbbm{k}_j^{(i)} \,, \quad 
    \mathbbm{k}^{(i)}_j = \id \otimes \ldots \id \otimes \overset{j}{\overset{\downarrow}{\mathbbm{k}^{(i)}}}\otimes \id \ldots \otimes \id \,,   
    \quad 
    \mathbbm{k}^{(i)} =  |i\rangle\langle i|\,,   
    \quad 
    \id = \sum_{m=0}^{d-1} |m\rangle\langle m| \,, 
    \quad i = 1, \ldots, d-1 \,.
    \label{Kidef}
\end{equation}
The Dicke states $|D^{n}(\vec{k})\rangle$ are simultaneous eigenstates of all these operators
\begin{equation}
\mathbb{K}^{(i)}\, |D^{n}(\vec{k})\rangle = k_i \,  |D^{n}(\vec{k})\rangle \,,
\qquad i = 1, \ldots, d-1 \,.
\label{U1sud}
\end{equation}
We define the corresponding unitary operators
\begin{equation}
    U^{(i)} = \exp \left( 2 \pi i \mathbb{K}^{(i)}/2^\ell\right) 
    = \prod_{j=1}^n \exp \left( 2 \pi i \mathbbm{k}^{(i)}_j/2^\ell\right) \,,
    \qquad i = 1, \ldots, d-1 \,,
    \label{QPEUsud}
\end{equation}
where $\ell$ is still to be determined (see \eqref{ellvaluesud} below), of which the Dicke states are simultaneous eigenstates
\begin{equation}
U^{(i)}\, |D^{n}(\vec{k})\rangle = e^{2 \pi i k_i/2^\ell} \,  |D^{n}(\vec{k})\rangle \,,
\qquad i = 1, \ldots, d-1 \,.
\end{equation}
The QPE circuit (discussed in Sec. \ref{sec:sudlogdepth}
below) uses the unitary operators \eqref{QPEUsud} to
project the product state \eqref{qpesudprod} to the Dicke state 
$|D^{n}(\vec{k})\rangle$ with a probability $P(\vec{k})$ given by
\begin{equation}
    P(\vec{k})= \frac{1}{(\vec{\xi} \cdot \vec{\xi})^n}   \left(\prod_{i=0}^{d-1}\xi_i^{2k_i}\right) \binom{n}{\vec{k}} \,,
\end{equation}
which is maximized for 
\begin{equation}
    \xi_i=\sqrt{\frac{k_i}{n}} \,, 
    \qquad i = 0, 1, \ldots. d-1 \,.
    \label{qpe_sudxi}
\end{equation}
The success probability of preparing $|D^{n}(\vec{k})\rangle$ is therefore given by
\begin{equation}
    P(\vec{k}) = \frac{n!}{n^{n}} \prod_{i=0}^{d-1} \frac{k_i^{k_i}}{k_i !}
    \approx \sqrt{\frac{n}{(2\pi)^{d-1} \prod_{\substack{i=0\\
    k_i \ne 0}}^{d-1}k_i}} \,.
    \label{sudprob}
\end{equation}
In the worst case $\vec{k}=(\frac{n}{d},\frac{n}{d},\ldots,\frac{n}{d})$, the number of required repetitions is $1/P(\vec{k}) \approx \mathcal{O}(n^{(d-1)/2})$. For typical $\vec{k}$-values, significantly fewer repetitions are needed.
For example, for the case \eqref{better}, $1/P(\vec{k}) \approx \mathcal{O}(x^{r/2})$.

\subsubsection{Log depth}\label{sec:sudlogdepth}

The circuit diagram for the standard QPE approach is shown in Figure \ref{fig:SU(d)-qpe}. The bottom wire represents the $n$-qudit product state \eqref{qpesudprod}. There are $(d-1)\ell$ qubit ancillas, where $\ell$ is the number of bits of $n$ (recall that $0 \le  k_i \le n$), namely,
\begin{equation}
    \ell=\lceil\log_2(n+1)\rceil \,.
    \label{ellvaluesud}
\end{equation}
The controlled unitaries are controlled versions of the unitary operators \eqref{QPEUsud}. The state of the system just prior to the measurement is
\begin{equation}
    \sum_{\substack{k_i =0, 1, \ldots, n \\
k_0 + k_1 + \ldots + k_{d-1} = n}} \sqrt{P(\vec{k})}
\, |D^n(\vec k)\rangle |k_{d-1} \ldots k_1\rangle \,,
\end{equation}
where $P(\vec{k})$ is given by \eqref{sudprob}. The circuit therefore succeeds on measuring the ancilla qubits’ base-10 values to be those of the target $\vec{k}$.

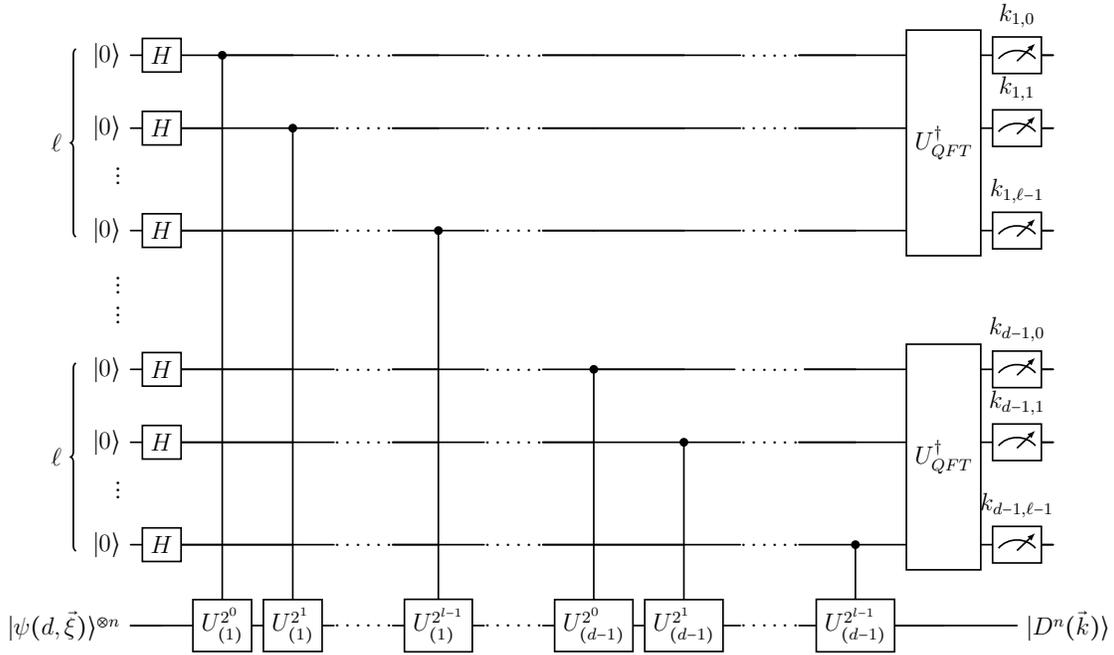
\begin{figure}[H]
    \centering
\begin{adjustbox}{width=0.9\textwidth}
  \begin{quantikz}[row sep=0.5cm, column sep=0.2cm]
  \lstick[4]{$\ell$} 
        &\wireoverride{n} 
        &\wireoverride{n}
         &\wireoverride{n}
         &\wireoverride{n}
    \lstick{$|0\rangle$}       
    & \gate{H}     
    & \ctrl{10}                   
    & \qw  
    & \ldots \ldots    
    & \qw     
    & \ldots \ldots            
    & \qw             
    & \qw  
    & \ldots \ldots  
    & \qw       
    & \gate[wires=4]{U_{QFT}^{\dagger}} 
    & \meter{k_{1,0}} 
    & \qw \\
     &\wireoverride{n} 
        &\wireoverride{n}
         &\wireoverride{n}
         &\wireoverride{n}
    \lstick{$|0\rangle$}       
    & \gate{H}     
    & \qw            
    & \ctrl{9}          
    & \ldots \ldots    
    & \qw  
    & \ldots \ldots                   
    & \qw 
    & \qw 
    & \ldots \ldots     
    & \qw           
    &                                
    & \meter{k_{1,1}} 
    & \qw \\
     &\wireoverride{n} 
        &\wireoverride{n}
         &\wireoverride{n}
         &\wireoverride{n}
    \lstick{$\vdots$}  \\
     &\wireoverride{n} 
        &\wireoverride{n}
         &\wireoverride{n}
         &\wireoverride{n}
    \lstick{$|0\rangle$}   
    & \gate{H}     
    & \qw               
    & \qw 
    & \ldots \ldots     
    & \ctrl{7}   
    & \ldots \ldots             
    & \qw            
    & \qw  
    & \ldots \ldots  
    &\qw     
    &                                
    & \meter{k_{1,\ell-1}} 
    & \qw \\
     &\wireoverride{n} 
        &\wireoverride{n}
         &\wireoverride{n}
         &\wireoverride{n}
    \lstick{$\vdots$}  \\ 
     &\wireoverride{n} 
        &\wireoverride{n}
         &\wireoverride{n}
         &\wireoverride{n}
    \lstick{$\vdots$}  \\
   \lstick[4]{$\ell$} 
        &\wireoverride{n} 
        &\wireoverride{n} 
         &\wireoverride{n}
         &\wireoverride{n}
    \lstick{$|0\rangle$}       
    & \gate{H}  
    & \qw            
    & \qw      
    & \ldots \ldots    
    & \qw  
    & \ldots \ldots  
    & \ctrl{4}   
    & \qw             
    & \ldots \ldots      
    & \qw                
    & \gate[wires=4]{U_{QFT}^{\dagger}} 
    & \meter{k_{d-1,0}}
    & \qw \\
     &\wireoverride{n} 
        &\wireoverride{n} 
         &\wireoverride{n}
         &\wireoverride{n}
    \lstick{$|0\rangle$}    
    & \gate{H}  
    & \qw  
    & \qw                
    & \ldots \ldots  
    & \qw  
    & \ldots \ldots  
    & \qw 
    & \ctrl{3}               
    & \ldots \ldots  
    & \qw 
    &                                
    & \meter{k_{d-1,1}} 
    & \qw \\
     &\wireoverride{n} 
        &\wireoverride{n}  
         &\wireoverride{n}
         &\wireoverride{n}
    \lstick{$\vdots$}    \\
     &\wireoverride{n} 
        &\wireoverride{n} 
         &\wireoverride{n}
         &\wireoverride{n}
    \lstick{$|0\rangle$}   
    & \gate{H}  
    & \qw               
    & \qw        
    & \ldots \ldots 
    & \qw    
    & \ldots \ldots            
    & \qw            
    & \qw  
    & \ldots \ldots  
    & \ctrl{1}      
    &                                
    & \meter{\quad k_{d-1,\ell-1}} 
    & \qw \\
     &\wireoverride{n} 
        &\wireoverride{n}  
         &\wireoverride{n}
         &\wireoverride{n}
    \lstick{$|\psi(d, \vec{\xi})\rangle^{\otimes n}$}      
    & \qw         
    & \gate{U^{2^0}_{(1)}} 
    & \gate{U^{2^1}_{(1)}}    
    & \ldots \ldots      
    & \gate{U^{2^{l-1}}_{(1)}} 
    & \ldots \ldots 
    & \gate{U^{2^0}_{(d-1)}} 
    & \gate{U^{2^1}_{(d-1)}}  
    & \ldots \ldots       
    & \gate{U^{2^{l-1}}_{(d-1)}} 
    & \qw                             
    & \rstick{$|D^{n}(\vec{k})\rangle$}
\end{quantikz}
\end{adjustbox}
    \caption{Circuit diagram for preparing the state $|D^{n}(\vec{k})\rangle$ in $\log$ depth using the standard QPE algorithm. All ancilla wires are qubits. The initial state of the bottom wire is \eqref{qpesudprod}, and 
    $U^{(i)}$ is defined in \eqref{QPEUsud}.}
    \label{fig:SU(d)-qpe}
\end{figure}
 
This circuit is evidently similar to the $SU(2)$ spin-$s$ version in Figure \ref{fig:log-spin-s}, differing mainly in the number of ancillas: the latter has only $\ell$, while the former has $(d-1)\ell$ in order to access all components of 
$\vec{k}$. Each of the $(d-1)\ell$ controlled unitaries can be implemented in constant depth using measurement/feedforward and $n$ additional qubit ancillas, see Result 1 in \cite{Piroli:2024ckr}. Hence, the controlled unitaries can be implemented in depth 
$\mathcal{O}((d-1)\ell)$. Each inverse quantum Fourier transform $U_{\rm QFT}^\dagger$ can be implemented in depth $\mathcal{O}(\ell)$. 
We therefore have the following:

\hypertarget{res: 6}{\noindent{\bf Result 6.}} The state $|D^{n}(\vec{k})\rangle$ can be prepared probabilistically with at worst
$\mathcal{O}(n^{(d-1)/2})$ repetitions and
with depth $\mathcal{O}(d \ell) = \mathcal{O}(d \log n)$, using
$\mathcal{O}(d \ell + n) = \mathcal{O}( d \log n + n)$ qubit ancillas. 

\subsubsection{Constant depth}\label{sec:sudconstantdepth}

Similarly to the case of $SU(2)$ spin-$s$ Dicke states discussed in Sec. \ref{sec:spinsconstantdepth}, variations of the above circuit can prepare the state $|D^{n}(\vec{k})\rangle$ in constant depth, at the cost of introducing additional and/or higher-dimensional ancillas.

The circuit in Figure \ref{fig:simple-const-SU(d)} consists of 
$(d-1)$ separate Hadamard tests using auxiliary qudits of dimension $\mathfrak{d}=n+1$, which is the number of possible values for each $k_i$. The controlled-$\mathcal{U}^{(i)}$ gates are defined,
for $i = 1, \ldots, d-1$, as
\begin{equation}
    C\mathcal{U}^{(i)}\, |y\rangle |x\rangle  
    = \left( \mathcal{U}^{(i)}(x) |y\rangle \right) |x\rangle \,, \quad
    \mathcal{U}^{(i)}(x) = \exp\left( 2 \pi i x \mathbb{K}^{(i)}/\mathfrak{d} \right) 
     = \prod_{j=1}^n \exp \left( 2 \pi i x \mathbbm{k}^{(i)}_j/\mathfrak{d} \right) \,, 
    \label{CcalU}
\end{equation}
where $\mathbb{K}^{(i)}$ and $\mathbbm{k}^{(i)}_j$
are defined in \eqref{Kidef}. These gates can be implemented in constant depth using measurement/feedforward and $n$ additional ancilla qudits of dimension $\mathfrak{d}$, by a generalization of the proof of Result 1 in \cite{Piroli:2024ckr}, see also appendix A in \cite{Zi:2025dgw}. 
We therefore have the following:

\hypertarget{res: 7}{\noindent{\bf Result 7.}} The state $|D^{n}(\vec{k})\rangle$ can be prepared probabilistically with at worst
$\mathcal{O}(n^{(d-1)/2})$ repetitions and
with depth $\mathcal{O}(d)$ (independent of $n$), using $\mathcal{O}(n+d)$ ancillas of dimension $n+1$. 

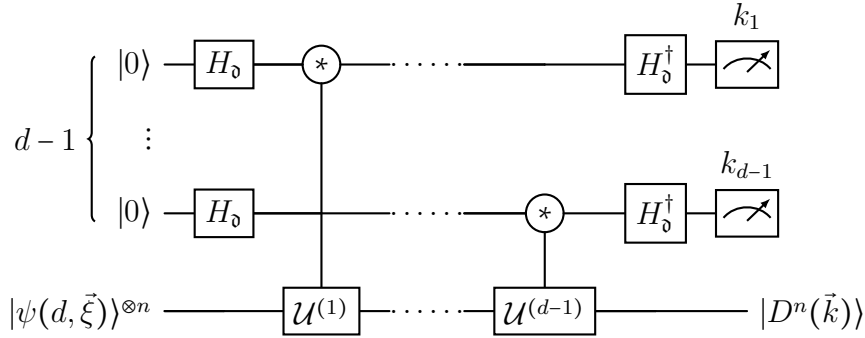
\begin{figure}[H]
    \centering
    \begin{quantikz}[row sep=0.6cm, column sep=0.4cm]
    \lstick[3]{$d-1$} 
        &\wireoverride{n} 
        &\wireoverride{n}
    \lstick{$|0\rangle$} 
    & \gate{H_\mathfrak{d}} 
    & \gateO{*} \vqw{3} 
    & \ldots \ldots  
    & \qw  
    & \gate{H_\mathfrak{d}^{\dagger}} 
    & \meter{k_1} \\
     &\wireoverride{n} 
        &\wireoverride{n}
    \lstick{$\vdots$}      \\
     &\wireoverride{n} 
        &\wireoverride{n}
    \lstick{$|0\rangle$} 
    & \gate{H_\mathfrak{d}} 
    & \qw  
    & \ldots \ldots 
    & \gateO{*} 
    \vqw{1} 
    & \gate{H_\mathfrak{d}^{\dagger}} 
    & \meter{k_{d-1}} \\
     &\wireoverride{n} 
        &\wireoverride{n}
    \lstick{$|\psi(d, \vec{\xi})\rangle^{\otimes n}$} 
    & \qw 
    & \gate{\mathcal{U}^{(1)}} 
    & \ldots \ldots 
    & \gate{\mathcal{U}^{(d-1)}} 
    & \qw  
    & \rstick{$|D^{n}(\vec{k})\rangle$}
\end{quantikz}
    \caption{Circuit diagram for preparing the state $|D^{n}(\vec{k})\rangle$ in constant depth using Hadamard tests. All ancilla wires are qudits of dimension $\mathfrak{d}=n+1$. 
    The initial state of the bottom wire is \eqref{qpesudprod}, and $\mathcal{U}^{(i)}$ is defined in \eqref{CcalU}.}
    \label{fig:simple-const-SU(d)}
\end{figure}

\par Finally, we can formulate an alternative constant-depth $SU(d)$ circuit
using ancillas of dimensions 2 and $d$,
by generalizing the $SU(2)$ spin-$s$ circuit in Figure \ref{fig:O(1)-spi}. Similarly to \eqref{initialalt}, the initial product state \eqref{qpesudprod} can be re-expressed as a superposition of computational basis states $|w\rangle$
\begin{equation}
    |\psi(d, \vec{\xi})\rangle^{\otimes n} =  
    \sum_{\substack{k'_i =0, 1, \ldots, n \\
k'_0 + k'_1 + \ldots + k'_{d-1} = n}}
    \sum_{w\in\mathfrak{S}_{M(\vec k')}}\, \alpha_{\vec{k'}, w}\, |w\rangle \,,
    \qquad \alpha_{\vec{k'}, w} 
    = \frac{\sqrt{P(\vec{k'})}}{ \sqrt{\binom{n}{\vec k'}}} \,,
\end{equation}
where the multiset $M(\vec k')$ is defined in \eqref{multiset}.
We fan out $(d-2)\ell + (\ell-1) =  (d-1)\ell -1$ times the state $|w\rangle$
(using $n ( (d-1)\ell  -1)$ qudit ancillas of dimension $n$, and corresponding qudit fan-out gates, denoted by $F$ in Figure \ref{fig:combined}) to obtain the state
\begin{equation}
    \sum_{\substack{k'_i =0, 1, \ldots, n \\
k'_0 + k'_1 + \ldots + k'_{d-1} = n}}
    \sum_{w\in\mathfrak{S}_{M(\vec k')}}\, \alpha_{\vec{k'}, w}\, |w\rangle^{\otimes (d-1)\ell} \,.
\end{equation}
We then use $(d-1)\ell$ qubit ancillas to apply a product of controlled gates $V^i$ defined by
\begin{equation}
    V^i(x)=\left( |0\rangle\langle 0| \right) \otimes\mathbb{I}
    +\left(|1\rangle\langle 1| \right) 
    \otimes U^i(x) \,,
    \quad U^i(x) = \exp\left( 2\pi i(\mathbb{K}^{(i)}-k_i)/2^x \right)\,, \quad  i = 1, \ldots, d-1 \,,
    \label{Videf}
\end{equation}
where $\vec{k}$ is the $\vec{k}$-value 
of the target state $|D^{n}(\vec{k})\rangle $.
With the help of the identity \eqref{identitytheta}, 
one can see that the state of the system just prior to measurement is 
\begin{equation}
    \sqrt{P(\vec{k})}\, |\underline{0}\rangle^{\otimes n ( (d-1)\ell  -1)}|D^{n}(\vec{k})\rangle |0\rangle^{\otimes (d-1)\ell}\, + \ldots \,.
\end{equation}
The circuit therefore succeeds on measuring all $(d-1)\ell$ qubit ancillas to be zero.

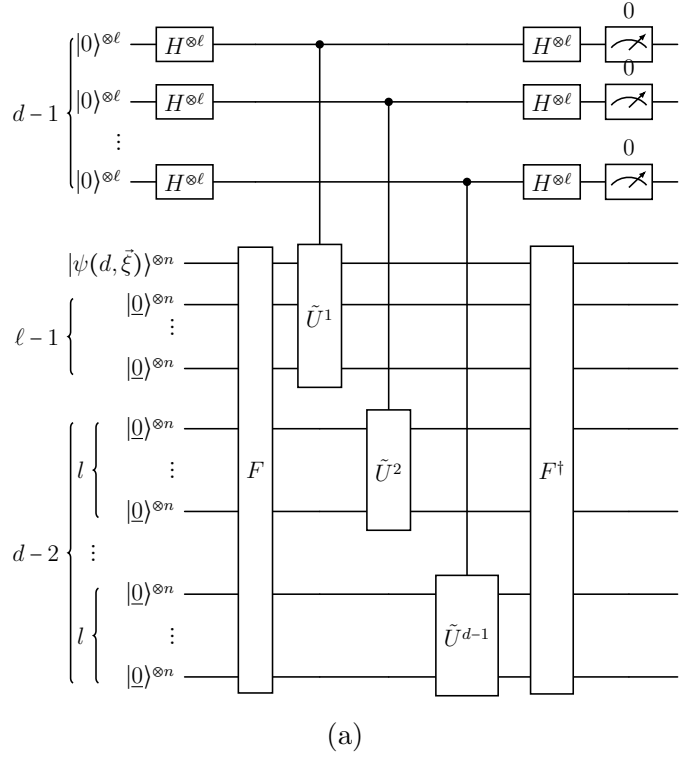
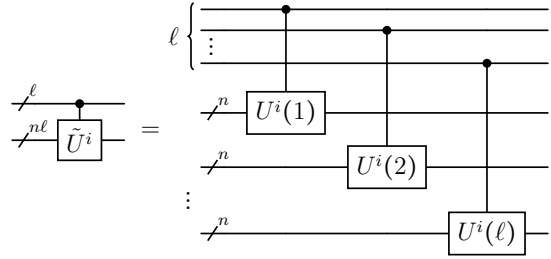
\begin{figure}[H]
\centering

\begin{subfigure}[b]{\textwidth}
\centering
\begin{adjustbox}{scale=0.75,center}

    \begin{quantikz}[row sep=0.4cm, column sep=0.45cm]
    \lstick[4]{$d-1$} 
        &\wireoverride{n} 
        &\wireoverride{n} \lstick{$|0\rangle^{\otimes \ell}$} &
        \gate{H^{\otimes \ell}} && \ctrl{8} 
        &&
        & \gate{H^{\otimes \ell}} 
        & \meter{0} 
        &\\
        &\wireoverride{n} &\wireoverride{n} \lstick{$|0\rangle^{\otimes \ell}$} &
        \gate{H^{\otimes \ell}} &&& \ctrl{8} 
        & 
        & \gate{H^{\otimes \ell}} 
        & \meter{0} 
        & \\
        &\wireoverride{n} &\wireoverride{n} \lstick{\vdots} \\
         &\wireoverride{n} &\wireoverride{n} \lstick{$|0\rangle^{\otimes \ell}$} &
        \gate{H^{\otimes \ell}} &&&
        & \ctrl{10} 
        & \gate{H^{\otimes \ell}} 
        & \meter{0} 
        & \\
        \\  &\wireoverride{n} &\wireoverride{n} &\wireoverride{n} \lstick{$|\psi(d, \vec{\xi})\rangle^{\otimes n}$} &\gate[11]{F} & \gate[4]{\tilde{U}^{1}} &&& \gate[11]{F^{\dagger}} && \rstick{$|D^n(\vec k)\rangle$}
        \\
         \lstick[3]{$\ell-1$} & \wireoverride{n} &\wireoverride{n} &\wireoverride{n} \lstick{$|\underline{0}\rangle ^{\otimes n}$} &&&&&&& \\
        &\wireoverride{n} &\wireoverride{n} &\wireoverride{n} \lstick{\vdots} \\
        &\wireoverride{n}  &\wireoverride{n} &\wireoverride{n} \lstick{$|\underline{0}\rangle^{\otimes n}$} &&&&&& & \\
        \lstick[7]{$d-2$} 
        &\wireoverride{n}   \lstick[3]{$l$} &\wireoverride{n} &\wireoverride{n} \lstick{$|\underline{0}\rangle^{\otimes n}$}
        &&& \gate[3]{\tilde{U}^{2}} &&&&\\   &\wireoverride{n}   &\wireoverride{n}     &\wireoverride{n} 
        \lstick{\vdots} \\   &\wireoverride{n} &\wireoverride{n}  &\wireoverride{n} \lstick{$|\underline{0}\rangle^{\otimes n}$}
         &&&&&&&
        \\
        &\wireoverride{n} \lstick{\vdots} \\ 
          &\wireoverride{n}   \lstick[3]{$l$} &\wireoverride{n} &\wireoverride{n} \lstick{$|\underline{0}\rangle^{\otimes n}$}
        &&&&\gate[3]{\tilde{U}^{d-1}} &&&\\   &\wireoverride{n}   &\wireoverride{n}     &\wireoverride{n} 
        \lstick{\vdots} \\   &\wireoverride{n}  &\wireoverride{n} &\wireoverride{n} \lstick{$|\underline{0}\rangle^{\otimes n}$}
         &&&&&&&
    \end{quantikz}

\end{adjustbox}

\caption{}
\label{fig:main-circuit}
\end{subfigure}

\vspace{0.5cm}

\begin{subfigure}[b]{\textwidth}
\centering
\begin{adjustbox}{scale=0.8,center}
 \begin{quantikz}[row sep=0.2cm, column sep=0.38cm]
     & \qwbundle{\ell} & \ctrl{1} & \\
     & \qwbundle{n \ell}&\gate{\tilde{U}^{i}} &
    \end{quantikz}
    =\begin{quantikz}[row sep=0.2cm, column sep=0.38cm]
  \lstick[4]{$\ell$} & & \ctrl{5} &&& \\
   & && \ctrl{5} && \\ &\wireoverride{n}
    \lstick{\vdots} \\
    &&&& \ctrl{5} & \\
    \lstick{}\\
    & \qwbundle{n} & \gate{U^i(1)} &&& \\
   & \qwbundle{n} && \gate{U^i(2)} && \\
    \lstick{\vdots} \\
   & \qwbundle{n} & && \gate{U^i(\ell)} &
    \end{quantikz}
\end{adjustbox}

\caption{}
\label{fig:subcircuit}
\end{subfigure}

\caption{(a) Circuit diagram for preparing the state 
$|D^{n}(\vec{k})\rangle$ in constant depth. Each of the top $(d-1)$ wires represent $\ell$ qubits, while each of the other wires represent $n$ qudits of dimension $d$. $F$ is a fan-out gate.
(b) Decomposition of the $\tilde{U}^i(x)$ sub-circuit, where $U^i(x)$ is defined in \eqref{Videf}.}
\label{fig:combined}
\end{figure}

The qudit fan-out gates can be implemented in
constant depth (see appendix A in \cite{Zi:2025dgw}), and likewise for the $V^i(x)$ gates (see Result 1 in \cite{Piroli:2024ckr}).
We therefore have the following:

\hypertarget{res: 8}{\noindent{\bf Result 8.}} The state $|D^{n}(\vec{k})\rangle$ can be prepared probabilistically with at worst
$\mathcal{O}(n^{(d-1)/2})$ repetitions and with depth $\mathcal{O}(1)$, using $\mathcal{O}(d \ell) = \mathcal{O}(d \log n)$ ancillas of dimension 2, and
$\mathcal{O}(n d \ell) = \mathcal{O}(n d \log n)$ ancillas of dimension $d$.

\section{Discussion}\label{sec:discussion}

We have presented a number of new ways of preparing qudit Dicke states. The circuits are explicit and straightforward, and are arguably simpler than those previously reported. (Implementations in cirq \cite{cirq} of all the circuits are available on GitHub \cite{GitHub}.)
Indeed, for $SU(2)$ spin-$s$ Dicke states, the sequential preparation circuits in Sec. \ref{sec:spinsDickeSeq} do not require separate treatment of ``edge'' cases, and do not require double-controlled gates, as does the circuit in \cite{Nepomechie:2024fhm}; and the corresponding QPE circuits in Sec. \ref{sec:spinsDickeQPE} are even simpler and
have lower depth, albeit at the expense of using additional and/or higher-dimensional ancillas and requiring multiple repetitions. For $SU(d)$ Dicke states, the comparison of the results in Secs. \ref{sec:sudDickeSeq} and \ref{sec:sudDickeQPE} with that of \cite{Nepomechie:2023lge} is similar. The corresponding algorithm in \cite{Liu:2024taj} has a superior blend of depth, ancillas and repetitions (see Table \ref{table:summary}), but is considerably more complicated. 

For the sequential preparation circuits, it would be in interesting to see if mid-circuit measurement and feedforward (local operations and classical communication, or LOCC) could be used to reduce circuit depth, as has been achieved for the preparation of multiqubit states, see e.g. \cite{Buhrman:2023rft, Piroli:2024ckr, Zi:2025dgw, Smith:2022nbd, Baumer:2023vrf, Baumer:2024jng, Yeo:2025tph}. 

A feature of the circuits in \cite{Nepomechie:2024fhm, Nepomechie:2023lge } is that they can straightforwardly prepare superpositions of Dicke states, following Theorem 2 in \cite{Bartschi2019}. A separate algorithm for preparing superpositions of $SU(d)$ Dicke states is also presented in  \cite{Liu:2024taj}. However, the circuits presented here will require modification and/or additional overhead in order to prepare such superpositions. For example, starting from the $SU(2)$ spin-$s$ circuit using the Hadamard test in Figure \ref{fig:simple-const-spin-s}, one can add a qudit ancilla encoding the amplitudes of the target superposition of Dicke states (which is ultimately measured),
and add suitable gates encoding the corresponding $k$-values; however, the success probability of preparing the target superposition will be smaller than for the original circuit, due to the measurement of the additional ancilla.

\section*{Acknowledgments}

We thank David Raveh for valuable discussions and comments on a preliminary draft.  
NK was supported by the National Science Foundation under grant 2244126; RN is supported in part by the National Science Foundation under grant PHY 2310594, and by a Cooper fellowship.

\appendix

\section{Proof of Eq. \eqref{proved}}\label{sec:proof}

We provide here a proof of Eq. \eqref{proved}, which is equivalent to the Proposition below. We use the notation introduced in Sec. \ref{sec:sudDickeSeq}. Our proof makes use of the following lemma. 

\begin{lemma*}
If $\vec{y} \in \mathcal{A}^{i+1}(\vec{k})$ and $\vec{y}>_{\rm lex} \vec{a} + \hat{0}$ for some $\vec{a} \in \mathcal{A}^{i}(\vec{k})$, then $\vec{y} = \vec{x} + \hat{0}$ for some $\vec{x} \in \mathcal{A}^{i}(\vec{k})$.
\end{lemma*}

\begin{proof} 
We prove the contraposition. That is, we show that the assumption $\vec{y} \ne \vec{x} + \hat{0}$ leads to a contradiction of the premise. Indeed, let us suppose that 
\begin{equation} 
\vec{y} \ne \vec{x} + \hat{0} = (x_0+1, x_1, \ldots, x_{d-1})\quad \text{where}\quad  0 \le x_0 \le k_0-1 \,; 0 \le x_r \le k_r\,, \quad  r = 1, \ldots d-1\,; \quad  \sum_{r=0}^{d-1}x_r=i \,.
\end{equation}
Then the zeroth component of $\vec{y}$ must be zero, i.e. $y_0 = 0$. This implies that 
\begin{equation} 
(a_0+1, a_1, \ldots, a_{d-1}) >_{\rm lex} (0, y_1, \ldots, y_{d-1}) \quad \text{for any}\quad \vec{a} \in \mathcal{A}^{i}(\vec{k}) \,.
\end{equation}
In other words, 
$\vec{a} + \hat{0} >_{\rm lex} \vec{y} $, which contradicts the premise.
\end{proof}

We are now ready to prove the following proposition.

\begin{prop*}
Let $\vec{a} \in \mathcal{A}^{i}(\vec{k})$ such that 
$\vec{a} + \hat{0} \in \mathcal{A}^{i+1}(\vec{k})$. Then the indices of $\vec{a}+\hat{0}$ and $\vec{a}$ satisfy
\begin{equation}
    J^{i+1}(\vec{a}+\hat{0}) \le J^{i}(\vec{a}) \,.
    \label{prop}
\end{equation}
\end{prop*}

\begin{proof}
Define $p:=J^{i}(\vec{a})$. Then there exist $p$ elements $\vec{x}^{(1)}, \ldots, \vec{x}^{(p)} \in \mathcal{A}^{i}(\vec{k})$ such that 
\begin{equation}
\vec{x}^{(r)} >_{\rm lex} \vec{a}\,, \quad r = 1, \ldots, p \,,
\end{equation}
and therefore
\begin{equation}
\vec{x}^{(r)} + \hat{0} >_{\rm lex} \vec{a} + \hat{0}\,, \quad r = 1, \ldots, p \,.
\label{inequality}
\end{equation}
Moreover, the Lemma implies that there are no additional elements
$\vec{y} \in \mathcal{A}^{i+1}(\vec{k})$ (beyond those elements in \eqref{inequality}) that satisfy $\vec{y}>_{\rm lex} \vec{a} + \hat{0}$. We conclude that
\begin{equation}
J^{i+1}(\vec{a}+\hat{0}) \le p \,,
\label{desired}
\end{equation}
as desired. 
\end{proof}

We note that the equality holds in \eqref{desired}
if $\vec{x}^{(r)} + \hat{0} \in \mathcal{A}^{i+1}(\vec{k})$ for all values of $r$, which occurs if $i+1 \le k_0$. Indeed, in this case, $x^{(r)}_0 \le i$ (since $\vec{x}^{(r)} \in \mathcal{A}^{i}(\vec{k})$ implies that $\text{max}(x^{(r)}_0) = \text{min}(k_0, i) = i$), and therefore $\vec{x}^{(r)} + \hat{0}  \in \mathcal{A}^{i+1}(\vec{k})$.

On the other hand, if $i+1 > k_0$, then there exists an element $\vec{\xi} \in \mathcal{A}^{i}(\vec{k})$ defined by
\begin{equation}
\vec{\xi} := (k_0, \text{max}(k_1, i-k_0),  \text{max}(k_2, i-k_0-k_1),\ldots, \text{max}(k_{d-1},i-\sum_{l=0}^{d-2}k_l) 
\end{equation}
that has 0 index, i.e. $J^i(\vec{\xi})=0$. Hence, 
$\vec{\xi} >_{\rm lex} \vec{a}$ (since $\vec{a} + \hat{0} \in \mathcal{A}^{i+1}(\vec{k})$ implies that $a_0<k_0$, while $\xi_0 = k_0$) and $\vec{\xi} + \hat{0} \notin \mathcal{A}^{i+1}(\vec{k})$. Referring again to \eqref{inequality}, we see that
$\vec{x}^{(r)} + \hat{0} \notin \mathcal{A}^{i+1}(\vec{k})$ for at least one value of $r$, which leads to a strict inequality in \eqref{desired}.

\clearpage
%\bibliographystyle{utphys}
%\bibliography{refs}

\providecommand{\href}[2]{#2}\begingroup\raggedright\endgroup

\end{document}